\documentclass{article}

\usepackage[letterpaper, margin=1in]{geometry} 

\usepackage{graphicx} 
\usepackage{amsmath, amsfonts, amsthm, amssymb, enumerate, xcolor, fullpage, algorithmic, physics}
\usepackage[ruled]{algorithm2e}
\usepackage{bbold}
\usepackage{float}
\usepackage{hyperref}
\hypersetup{colorlinks,linkcolor={blue},citecolor={blue},urlcolor={red}} 
\usepackage{enumitem}

\usepackage[capitalize,nameinlink]{cleveref}
\crefrangemultiformat{equation}{#3Eqs.~((#1))#4--#5((#2))#6}{ and #3(#1)#4--#5(#2)#6}{, #3(#1)#4--#5(#2)#6}{ and #3(#1)#4--#5(#2)#6}

\newcommand{\eps}{\varepsilon}
\newcommand{\Be}{\operatorname{B}}
\newtheorem{theorem}{Theorem}[section]
\newtheorem{theoremstar}[theorem]{Theorem$^\dagger$}

\newtheorem{lemma}[theorem]{Lemma}

\newtheorem{corollary}[theorem]{Corollary}
\newtheorem{proposition}[theorem]{Proposition}

\newtheorem{observation}[theorem]{Observation}

\newtheorem{fact}[theorem]{Fact}
\theoremstyle{definition}
\newtheorem{definition}[theorem]{Definition}

\newcommand{\calC}{\mathcal{C}}
\newcommand{\calU}{\mathcal{U}}
\newcommand{\calM}{\mathcal{M}}
\newcommand{\calJ}{\mathcal{J}}
\newcommand{\calP}{\mathcal{P}}

\newcommand{\E}{\mathbb{E}}
\newcommand{\be}{\mathrm{B}}

\newcommand{\dt}{\,\mathrm{d}t}

 \usepackage{etoolbox}
  \makeatletter
    \patchcmd{\maketitle}{\@fnsymbol}{\@arabic}{}{}
    \patchcmd{\maketitle}{\setcounter{footnote}{0}}{}{}{}
  \makeatother

\begin{document}

\title{The Dirichlet Mechanism for rounding with strong negative correlation, with applications}
 \author{David G. Harris\footnote{University of Maryland. Email: {\tt davidgharris29@gmail.com}}, \thickspace \thickspace George Z. Li\footnote{Carnegie Mellon University. Email: {\tt gzli@andrew.cmu.edu}}, \thickspace \thickspace Nitya Raju\footnote{University of Maryland. Email: {\tt nraju@umd.edu}}, \thickspace \thickspace Renata Valieva\footnote{University of Maryland. Email: {\tt rvalieva@umd.edu}}}
\date{}

\maketitle

\begin{abstract}
Many optimization and scheduling problems can be abstracted in terms of a bipartite ``assignment graph" $G = (U \cup V, E)$, where the goal is to select exactly one edge for each right-node. For example, a right-node may correspond to a job, and a left-node to a possible machine assignment.  A common strategy to solve such problems is to obtain a fractional relaxation $x_e$ for each edge $e$, and then have each right-node independently select an edge with probability $x_e$. However, this may cause the left-nodes to become unevenly loaded, leading to suboptimal solutions for some problems.

To address this, a number of algorithms for dependent rounding with strong negative correlation have been developed, e.g. Bansal, Srinivasan \& Svensson (2021), Im \& Shadloo (2020), Im \& Li (2023), Harris (2025), Naor, Srinivasan \& Wajc (2025).   We introduce a new method for this, which we call the \emph{Dirichlet mechanism}. It is based on having each left-node draw Dirichlet random variables for its edges, and then having each right-node select an edge based on these values.  This achieves quantitatively stronger negative correlation than previous algorithms, and is also simpler since it avoids the need for a tie-breaking mechanism.

We illustrate the mechanism with improved approximation ratios for two problems. For oblivious online dependent rounding, we achieve a $0.68$-approximation which improves upon the previous $0.652$-approximation of Naor, Srinivasan \& Wajc (2025). For the  problem of scheduling jobs on unrelated machines to minimize weighted completion time, we achieve a $1.387$-approximation which improves upon the $1.398$-approximation of Harris (2025). (A recent algorithm of Li (2025)  based on iterated rounding also provides a $1.36$-approximation if the weights of each job are independent of machine.)
\end{abstract}

\section{Introduction}
Many scheduling and resource 
allocation problems can be formulated as a \emph{bipartite assignment problem}: we are given a complete bipartite graph $G = (U \cup V, E)$, and we wish to select a set of edges $K$ which intersects each right-node $v \in V$ exactly once. That is, $K$ is a ``half-matching". For instance, $V$ can represent a set of jobs to be scheduled, and $U$ can represent a set of possible machines. Alternatively, $V$ can represent a set of items to be sold, and $U$ can represent potential buyers.

There is a natural strategy for solving such problems. First, one solves a relaxation (e.g., an LP relaxation) to obtain a fractional solution $(x_e: e \in E)$. Then, for each right-node $v \in V$, one selects exactly one neighboring edge $e \in N(v)$, wherein each edge is selected with probability $x_e$. We refer to this as \emph{independent rounding.} (Here, $N(v)$ denotes the set of edges incident on vertex $v$). This algorithmic approach applies to a wide range of problems, and we do not need to belabor its power.

An inherent limitation of independent rounding is that the left-nodes (e.g. the machines in a scheduling problem) can become unevenly loaded due to random fluctuations in edge selections. This is alleviated when the average number of selected edges per left-node is large, but not all problems lie in the ``concentration threshold" regime.   

In a breakthrough result, \cite{bansal2016lift} devised a new rounding approach, based on \emph{dependent rounding with strong negative correlation.} Instead of right-nodes acting independently, they are tied together in a scheme wherein each edge $e$ is still \emph{marginally} selected with probability $x_e$, but for any pair of distinct edges $e,f$ incident on the same left node, there is a \emph{strong negative correlation,} i.e. 
$$
\E[X_e X_f] \leq (1 - c) x_e x_f
$$
for some small constant $c>0$, where $X_e, X_f$ are indicator variables for selecting edges $e,f$. We contrast this with mere \emph{negative correlation}, which ensures only that
$$
\E[X_e X_f] \leq x_e x_f
$$


 The strong negative correlation property implies that the loads on each left-node become more balanced. Using this approach, \cite{bansal2016lift} achieved a $(1.5 - \eps)$-approximation algorithm for a classical scheduling problem of minimizing weighted completion time on unrelated machines, for some minuscule constant $\eps > 0$. Notably, they showed that independent rounding --- even when given a solution $\vec x$ which is a convex combination of optimal integral solutions --- is inherently limited to a $1.5$-approximation ratio.

Since then, a variety of rounding schemes with strong negative correlation have been developed, leading to improved approximation ratios for this scheduling problem and others \cite{im2020,im2023,baveja2024,harris2024dependent}. Generally speaking, we can group these schemes into two classes. The first approach, as in the original work of \cite{bansal2016lift}, is based on a random walk: at each stage, the fractional vector $\vec x$ is modified, until eventually it becomes integral.  An improved version of this rounding scheme was later developed in \cite{baveja2024}.

The second approach, developed originally by \cite{im2023}, is based on ideas from \emph{contention-resolution} in economics. Here, each left-node ``bids" for the right-nodes, and these demands are balanced so that each left-node does not get too many edges. As a result, the variables $X_e, X_f$ become negatively correlated. The algorithm of \cite{im2023} was based on Poissonian ``tickets" for the allocation; the later work  \cite{harris2024dependent} was based on a multivariate geometric distribution.

\subsection{Our Contributions}

In this work, we develop a new randomized rounding method that falls firmly into the second approach, which we call the \emph{Dirichlet mechanism.} It is based on \emph{Dirichlet random variables} as the underlying probability distribution for contention among the left-nodes. Such random variables have powerful properties: they have continuous CDFs, are negatively associated, are closed under aggregation, can be simulated online, and so on.

The correlation function does not have a closed-form expression, and requires some grueling analysis of the Incomplete Beta function for general calculations. For many  algorithmic applications, as we will see, the worst-case behavior comes when the fractional relaxation $\vec x$ has infinitesimal entries. In this special case, we can summarize the new algorithm crisply:
\begin{theorem}[Simplified]
Suppose that distinct edges $e,f$ share a left-node $u$, and entries $x_e, x_f$ are infinitesimal. Then the dependent rounding ensures that\footnote{This holds by setting $\rho_e = x_e/a$ for all edges $e$; see \Cref{thm: E[U^q]} and \Cref{poisson-prop} for further details.}
$$
\E[X_e X_f] \leq \frac{x_e x_f}{\binom{2/a}{1/a}} \qquad \text{for $a = \sum_{e \in N(u)} x_e$}
$$
\end{theorem}

 By contrast, the rounding scheme in \cite{harris2024dependent} would give $\E[X_e X_f] \leq \frac{2 x_e x_f}{1 + e^{1/a}}$ in the infinitesimal setting. The rounding schemes of \cite{bansal2016lift,im2020,baveja2024} have not been analyzed for general values of $a$; for $a = 1$, they would give coefficients of $\frac{107}{108}, \frac{2}{\mathrm{e}}$ and $\frac{26}{27}$ respectively, which apply also for non-infinitesimal $x_e, x_f$.
 
At a high level, our approach is similar to \cite{harris2024dependent}: we construct random variables for each edge that are marginally continuous uniform random variables, but collectively have strong negative correlation. Instead of generating them via the Multivariate Geometric distribution, we use a Dirichlet distribution.
The new algorithm is quantitatively stronger. It is also simpler: because the Dirichlet distribution has a continuous CDF, there is no need for the tie-breaking mechanisms required in \cite{im2023} or \cite{harris2024dependent}.

By plugging our new dependent-rounding algorithm into the framework of \cite{harris2024dependent}, we obtain an improved approximation algorithm for the problem of \emph{Scheduling on Unrelated Machines to Minimize Weighted Completion Time}, essentially for free. In addition, our new rounding algorithm can be used for a novel, and much simpler, algorithm for the problem of \emph{Oblivious Online  Matching}.  We now provide an overview of these problems, and how our algorithm applies to each.

\subsection{Application: Oblivious Online Matching}
The online dependent rounding problem was introduced in \cite{srinivasan2023online}, as an abstraction of allocation problems with information about long-term average-case demand. In this scenario, we have a bipartite graph $G = (U \cup V, E)$, where the ``offline nodes" $U$ are fixed and known in advance. Each time an ``online" node $v \in V$ appears, we learn the demand values $g_{(u,v)} : u \in U$.  We are guaranteed that the overall collection of demands $g_e$ forms a fractional matching of $G$.  Upon the arrival of $v$ we must immediately and irrevocably select some edge $e$ incident to $v$ to add to a matching $M$ (or opt to add none). The goal is to achieve $$
\Pr(e \in M) \geq \gamma \thinspace g_e \qquad \text{for some factor $\gamma \in [0,1]$}
$$

In the offline setting, it is trivial to achieve ratio $\gamma = 1$, since the bipartite fractional matching polytope has no integrality gap. There is a simple online scheme to achieve ratio $\gamma = 1 - 1/e \approx 0.632$. The work \cite{srinivasan2023online} showed an online upper-bound of $\gamma \leq 0.828$, and provided a sophisticated algorithm with $\gamma = 0.652$. 

The Dirichlet mechanism can be implemented in the online setting, which we use to provide an improved algorithm for this problem. We get the following result:
\begin{theorem}
There is an algorithm for Oblivious Online Matching with ratio $\negthinspace \gamma = 0.68$.
\end{theorem}
Beyond the slightly improved factor, the algorithm is much simpler than \cite{srinivasan2023online}. Our full algorithm fits comfortably on half a page.  We emphasize that this is the first  new application of strongly negatively correlated randomized rounding schemes since its introduction in~\cite{bansal2016lift}.

\subsection{Application: Weighted Completion Time}
The classical scheduling problem of \emph{Weighted Completion Time on Unrelated Machines} is denoted in the scheduling literature as $R || \sum_j w_j C_j$. We have a set of machines $\mathcal M$ and a set of jobs $\mathcal J$, where each job $j$ has a weight $w_{j}$ and a separate processing time $p_{ij}$ on each machine $i$. The objective is to assign jobs to machines in some order, so as to minimize the total weighted completion time $\sum_j w_j C_j$, where $C_j$ is the cumulative processing time of all jobs assigned to machine $i$ up to and including $j$, $\sum_{j'\leq j}p_{ij'}$. There is also a variant where the weight of a job $j$ may depend on the machine $i$ to which it is assigned.

This problem has attracted attention, in
part, because it leads to sophisticated rounding algorithms. In particular, the objective function can be seen as a \emph{quadratic} function of the underlying assignment variables. Breaking a long-standing barrier, \cite{bansal2016lift} devised a rounding algorithm with an approximation ratio $1.5 - \eps$ for a very small constant $\eps > 0$, based on dependent rounding with negative correlation. 

Since then, there have been a number of improved approximation algorithms using various dependent rounding schemes. Most of these works can be summarized in the following framework: first solve a convex relaxation, next cluster the jobs by processing times and weights, and finally apply a rounding scheme with strong negative correlation within clusters. (The recent algorithm \cite{li2025approximating} is an exception to this pattern: it is more in the genre of iterated rounding.) We summarize these algorithms  as follows:

\begin{table}[H]
    \begin{center}
    \begin{tabular}{|c|c|c|c|}
    \hline
      Ref.   & Relaxation & Machine-varying & Ratio \\
      & & weights? & \\
      \hline
      \hline
    \cite{bansal2016lift}  & Semidefinite Program & $\checkmark$ & $1.5 - \eps$\\
\cite{li2020} & Time-indexed LP         & $\checkmark$ & $1.5 - \eps$ \\
\cite{im2020} & Time-indexed LP         & $\checkmark$ & $1.488$ \\
\cite{im2023}   & Time-indexed LP & $\checkmark$ & $1.45$ \\
\cite{harris2024dependent} & Semidefinite Program & $\checkmark$& $1.398$ \\
\cite{li2025approximating} & Configuration LP & \text{\sffamily X} & $1.36$ \\
\hline
This work & Semidefinite Program  & $\checkmark$ & $1.387^{\dagger}$ \\
\hline
    \end{tabular}
    \label{tab:my_label}
\end{center}
\end{table}

Plugging the Dirichlet mechanism into the algorithm of \cite{harris2024dependent} gives us a slightly stronger approximation ratio, essentially for free. (Since our analysis is very similar to \cite{harris2024dependent}, we used off-the-shelf numerical optimization methods instead of exact arithmetic for the computations. As a result, the numerical figure here is technically a pseudo-theorem, which we marked with a dagger.)
\section{Preliminaries}\label{sec: 02_Dirichlet}
In this section we review background on two key probabilistic notions used throughout the paper: (i) Negative Association of random variables, and (ii) the Dirichlet distribution and its properties.

At several points, we silently assume that relevant quantities are non-zero. For example, we may write expressions such as $\frac{1}{x} - 1$ without verifying that $x > 0$. In all cases, the formulas can be extended to the boundary case $x = 0$ in a straightforward way. This convention allows us to avoid cluttering the exposition with excessive edge cases.

\subsection{Negative Association}
We begin by reviewing the general properties of negatively associated (NA) random variables. 

\begin{definition}[Negatively Associated random variables]
    A finite collection of random variables $\chi = (X_1, X_2, \dots, X_n)$ is  \textit{negatively associated} (NA) if for every pair of disjoint subsets $A,B\subseteq[n]$, we have $$    \mathbb{C}\text{ov}(f(X_i, i \in A), g(X_j, j \in B)) \leq 0$$ for all functions $f$ and $g$ that are increasing in each argument.
\end{definition}

It is easy to see that if $\chi$ and $\chi'$ are collections of NA random variables, and $\chi$ is independent of $\chi'$, then their union $\chi \cup \chi'$ is also NA. We quote a few standard facts from \cite{Joagdev1983NegativeAO}.
\begin{theorem}\label{thm: NA properties}
Let \(X_1, \dots, X_k\) be NA random variables. Then:
\begin{enumerate}
    \item[(a)] If \(X_1, \dots, X_k\) are nonnegative, then \( \mathbb{E}[X_1 \cdots X_k] \leq \mathbb{E}[X_1] \cdots \mathbb{E}[X_k] \).
 
    \item[(b)] If \(f_1, \dots, f_\ell\) are functions defined on disjoint subsets of \(\{1, \dots, k\}\), where \(f_1, \dots, f_\ell\) are either all monotonically non-increasing or all monotonically non-decreasing, then the random variables \(f_i(\vec{X})\) for \(i = 1, \dots, \ell\) are NA. 
    \item[(c)] (Restatement of (a) + (b)) If  \(f_1, \dots, f_\ell\) are nonnegative functions defined on disjoint subsets of \(\{1, \dots, k\}\), where \(f_1, \dots, f_\ell\) are either all monotonically non-increasing or all monotonically non-decreasing, then
    $$
    \E \Bigl[ \prod_{i=1}^{\ell} f_i(\vec X) \Bigr] \leq \prod_{i=1}^{\ell} \E \Bigl[ f_i(\vec X) \Bigr]
    $$ 
\end{enumerate}
\end{theorem}

\subsection{The Dirichlet distribution}
Our algorithms center around a probability distribution known as the \emph{Dirichlet distribution}. In order to define it, we first introduce the \emph{Beta function}, 
$$
\Be(a,b) = \frac{\Gamma(a) \Gamma(b)}{\Gamma(a+b)} = \int_0^1 t^{a-1} (1-t)^{b-1}  \ \mathrm{d}t
$$
where $\Gamma$ is the Gamma function. This is closely related to the generalized Binomial coefficient defined by
$$
\binom{x}{y} := \frac{\Gamma(x+1)}{\Gamma(y+1) \Gamma(x-y+1)} = \frac{1}{\be(x-y, y+1) (x-y)}
$$

There are a number of important variants and extensions of the Beta function that will be useful for our analysis.  The first is the \emph{Incomplete Beta function}, given by 
\[\Be(z; a, b) = \int_0^z t^{a-1} (1-t)^{b-1} \ \mathrm{d} t\]
The \emph{regularized Incomplete Beta function} is defined as
\[I(z; a, b) = \frac{\Be(z; a, b)}{\Be(a, b)},\]
which is normalized so that $I(0;a,b) =0, I(1;a,b) = 1$.

Another variant is the \emph{multivariate} Beta function, defined for a vector $\vec \rho = (\rho_0, \dots, \rho_n)$ by
 \[
    \Be(\vec \rho) = \frac{\prod_{i=0}^n \Gamma(\rho_i)}{\Gamma\left(\sum_{i=0}^n \rho_i\right)}
    \]

\begin{definition}[Dirichlet distribution]
    The \textit{Dirichlet distribution} is a family of continuous multivariate probability distributions parameterized by a vector of positive reals \(\vec \rho = (\rho_1, \rho_2, \dots, \rho_n)\). It is denoted as \(\operatorname{Dir}(\vec\rho)\) and defined on the \( (n-1) \)-dimensional simplex by the probability density function
    \[
    f(x_1, \dots, x_n; \rho_1, \dots, \rho_n) = \frac{1}{\Be(\vec\rho)} \prod_{i=1}^n x_i^{\rho_i - 1}
    \]
    where \(x_i \geq 0\), \(\sum_{i=1}^n x_i = 1\).
\end{definition}

The Dirichlet distribution can be viewed as a distribution over vectors with \(x_1 + \dots + x_n \leq 1\), with PDF:
$$
 f(x_1, \dots, x_n; \rho_1, \dots, \rho_n, \rho_0) = \frac{1}{\Be(\vec\rho)} (1 - x_1 - \dots - x_n)^{\rho_0 - 1} \prod_{i=1}^n x_i^{\rho_i - 1} 
 $$
 for $\vec\rho = (\rho_1, \dots, \rho_n, \rho_0)$. This effectively adds a dummy variable \(x_0\) as the $(n+1)^{\text{st}}$ dimension to lie on the $n$-dimensional simplex.

\begin{fact}[Properties of the Dirichlet distribution]\label{fact:prop-of-dirichlet}
Let $(X_1, \dots, X_n) \sim \operatorname{Dir}(\rho_1, \dots, \rho_n)$.
\label{dirichletfact}
    \mbox{}\\
    \vspace{-0.2in}
    \begin{enumerate}
        \item \textbf{Marginals have a Beta distribution:} Each coordinate $i$ satisfies $X_i \sim \operatorname{Beta}(\rho_i, \sum_{j \neq i} \rho_j)$. Its CDF is given by
        $$
        \Pr(X_i \leq z) = I(z; \rho_i, \sum_{j \neq i} \rho_j) \qquad \text{for $z \in [0,1]$}.
        $$
        
\item \textbf{Symmetry:} If $\pi$ is a permutation of $\{1, \dots, n \}$, then $(X_{\pi 1}, \dots, X_{\pi n}) \sim \operatorname{Dir}(\rho_{\pi 1}, \dots, \rho_{\pi n})$.

        \item \textbf{Aggregation:} If two components $X_i$ and $X_j$ are aggregated by replacing  $X_i$ and $X_j$ with $X_i + X_j$, the new vector $X' = (X_1, \ldots, X_i + X_j, \ldots, X_n)$ also follows a Dirichlet distribution, 
        $$X' \sim \operatorname{Dir}(\rho_1, \ldots, \rho_i + \rho_j, \ldots, \rho_n).$$

        \item \textbf{Neutrality:} For any $k \leq n$, the random variable $X_k$ is independent of 
        $$
        \left(\frac{X_1}{1-X_k}, \dots, \frac{X_{k-1}}{1-X_k}, \frac{X_{k+1}}{1-X_k}, \dots, \frac{X_n}{1-X_k}\right).
        $$

        \item \textbf{Negative association:} $X_1, \dots, X_n$ are NA r.v.'s  
        \cite{barthe2010generalized}
        
        \item \textbf{Moments:} For a vector of nonnegative reals $\vec \beta = (\beta_1, \dots, \beta_n)$, there holds $$\E\left[\prod_{i=1}^{n}X_i^{\beta_i}\right] = \frac{\Be(\vec \rho + \vec \beta)}{\Be(\vec \rho)} = 
 \frac{\Gamma\left(\sum_{i=1}^{n}\rho_i\right)}{\Gamma \left( \sum_{i=1}^{n} (\rho_i+\beta_i) \right) } \prod_{i=1}^{n}\frac{\Gamma(\rho_i + \beta_i)}{\Gamma(\rho_i)}.$$ 
    \end{enumerate}
\end{fact}

\section{The Dirichlet mechanism}
\label{sec:dir}
We introduce the \emph{Dirichlet mechanism} for bipartite rounding, in two parts. The first part is the generation of a certain type of copula (i.e. a multivariate probability distribution whose marginals are all $\text{Unif}[0,1]$) with negative correlation properties.  This is achieved by transforming multivariate Dirichlet variables. The second part is a method of bipartite selection using these copulas. This rounding step is exactly the same as the method of \cite{harris2024dependent}; for completeness, we include the analysis in this section.

\subsection{The Dirichlet copula}
This step is deceptively simple: generate a multivariate Dirichlet vector and transform each of its marginals to a Uniform. Here we use \Cref{fact:prop-of-dirichlet}(1), which gives the explicit marginal CDF in terms of the regularized Incomplete Beta function $I(z;\rho_i;1-\rho_i)$.

\begin{algorithm}[H]
\caption{\textsc{DirichletCopula}($\vec{\rho}$)}
\begin{algorithmic}[1]
\REQUIRE $\vec{\rho}\in(0,1)^n$ with $\sum_{i=1}^n\rho_i \leq 1$
\STATE Draw vector $(T_1,\ldots,T_n)\sim\text{Dir}(\rho_1,\ldots,\rho_n, 1 - \rho_1 - \dots - \rho_n)$ from a Dirichlet distribution.  \thickspace \thickspace  \thickspace \thickspace \thickspace \thickspace \thickspace \thickspace \thickspace \thickspace \thickspace \thickspace \thickspace \thickspace \thickspace \thickspace \thickspace \thickspace (See e.g \cite{gelman2003bayesian} for a discussion of how to do this efficiently.)  \\

\FOR{$i=1,\ldots, n$}
    \STATE Compute $A_i\gets I(T_i;\rho_i,1-\rho_i)$
\ENDFOR
\RETURN $(A_1,\ldots, A_n)$
\end{algorithmic}\label{alg:correlated-uniform-rvs}
\end{algorithm}
 
\begin{observation}
For $(A_1, \dots, A_n) \sim \text{DirichletCopula}(\vec \rho)$, each variable $A_i$ follows the distribution $\text{Unif}([0,1])$
\end{observation}
\begin{proof}
By \Cref{dirichletfact}(1), each $T_i$ is marginally distributed as $\text{Beta}(\rho_i; 1 - \rho_i)$ with CDF $I(x, \rho_i, 1 - \rho_i)$, which is a continuous function of $x$. Thus, the mapping from $T_i$ to $A_i$ is the well-known probability integral transform for continuous random variables.
\end{proof}

\begin{observation}
\label{dirichlet-subsample}
Suppose $(A_1, \dots, A_n) \sim \text{DirichletCopula}(\vec \rho)$, and $i_1, \dots, i_k$ are any distinct indices in $[n]$. Then $(A_{i_1}, \dots, A_{i_k}) \sim \textsc{DirichletCopula}(\rho_{i_1}, \dots, \rho_{i_k})$.
\end{observation}
\begin{proof}
Follows from symmetry and aggregation properties of the Dirichlet.
\end{proof}

The $\vec \rho$ vector is somewhat mysterious; roughly speaking, $\rho_i$ controls the amount of ``anti-correlation" of $A_i$ against the other coordinates $A_j$. If $\rho_i = 0$, then $A_i$ is just an independent uniform. The following result makes this picture more precise.
\begin{proposition}[Concordance monotonicity in $\rho$]
\label{prop:dirichlet-concordance}
Let $\vec A, \vec A'$ be generated by
$\textsc{DirichletCopula}$ for parameter vectors $\vec \rho, \vec \rho'$ on $n$ coordinates with $\rho_i \leq \rho'_i$ for all $i \in [n]$.  Then
$$
\vec A' \preceq_{\mathrm c} \vec A
$$
where $\preceq_{\mathrm c}$ denotes precedence in the concordance order; that is, for every $\vec u \in [0,1]^n$,
\[
\Pr \Bigl( \bigwedge_{i} A'_i\leq u_i \Bigr)
\leq
\Pr \Bigl( \bigwedge_{i} A_i\leq u_i \Bigr)  \qquad \text{ and } \qquad
\Pr \Bigl( \bigwedge_{i} A'_i>u_i \Bigr) \leq 
\Pr \Bigl( \bigwedge_{i} A_i>u_i \Bigr)
\]
\end{proposition}
\begin{proof}
By transitivity, it suffices to consider a case where $\rho, \rho'$ differ in just a single coordinate; say $\rho_1 < \rho_1'$ and $\rho_i = \rho'_i$ for $i = 2, \dots, n$.  Using the neutrality and aggregation properties of the Dirichlet distribution, we consider the following coupling construction to simultaneously draw $ \vec X \sim \text{Dir}(\vec \rho)$ and $\vec X' \sim \text{Dir}(\vec \rho')$: 
 \begin{enumerate}
 \item First, draw $X_1' \sim \text{Beta}(\rho_1', 1 - \rho_1')$.
 \item Then independently draw $(Z_2, \dots, Z_n) \sim \text{Dir}(\rho_2, \dots, \rho_n, 1 - \rho'_1 - \rho_2 - \dots - \rho_n)$ and set $$
 (X_2, \dots, X_n) = (1 - X_1') (Z_2, \dots, Z_n).
 $$
 \item Finally, independently draw $Z_1 \sim \text{Beta}(\rho_1, \rho_1' - \rho_1)$ and set $X_1 = Z_1 X_1'$
 \end{enumerate}

Accordingly, define
\[
A_1=I(X_1; \rho_1, 1- \rho_1)
\qquad
A_1'=I(X_1'; \rho_1', 1 - \rho_1')
\qquad
A_i=A_i'=I(X_i; \rho_i, 1 - \rho_i) 
\quad(2\leq i\leq n).
\]
All the variables $A_i$ and $A_i'$ are uniform on $[0,1]$.

We first compare upper-orthant probabilities.  Fix
$\vec u \in[0,1]^n$ and define
\[
M_+(w)
=
\Pr \Bigl( \bigwedge_{i=2}^n A_i>u_i      \mid A_1'=w \Bigr)
\]

Since, in our coupling construction, each term $X_i: i \geq 2$ is a decreasing function of $X_1'$ for fixed $Z$, and since each
$I(z; \rho_i, 1 - \rho_i)$ is increasing, $M_+(w)$ is nonincreasing in $w$.

Put
\[
p_+(w)=\Pr(A_1>u_1\mid A_1'=w).
\]
Because $A_1$ and $A_1'$ are uniform,
\begin{equation}
\label{ppbound}
0\leq p_+(w)\leq1,
\qquad
\int_0^1p_+(w)\,\dd w= \Pr( A_1 > u_1 ) = 1-u_1.
\end{equation}

Since $A_1$ and $A_2, \dots, A_n$ are conditionally independent given $A_1'$ (equivalently, given $X_1'$) in our coupling, we have
\[
\Pr \Bigl( \bigwedge_{i=1}^n A_i>u_i \Bigr)
=
\int_0^1p_+(w)M_+(w)\,\dd w.
\]

By the Hardy-Littlewood  Rearrangement Inequality, this integral is minimized, under the conditions (\ref{ppbound}), when $p_{+}(w)$ puts all its mass on the upper range of the interval, i.e. with $p_+(w) = 1 \{w > u_1 \}$. Hence 
$$
\int_0^1p_+(w)M_+(w)\,\dd w \geq \int_{u_1}^1 M_+(w) \,\dd w = \Pr \Bigl( \bigwedge_{i=1}^n A_i'>u_i \Bigr)
$$

The lower-orthant comparison follows a symmetric monotonicity and rearrangement argument, where we define
\[
M_-(w)
=
\Pr \Bigl( \bigwedge_{i=2}^n A_i\leq u_i \mid A_1' = w \Bigr), \qquad \text{and}  \qquad
p_-(w)=\Pr(A_1\leq u_1\mid A_1'=w). \qedhere
\]
\end{proof}

Most of the analysis can be stated in terms of a function $\Psi$, which measures the correlation between the random variables $A_i$ and $A_j$. Critically, because of \Cref{dirichlet-subsample}, this only depends on the values $\rho_i,  \rho_j, x_i, x_j$ (and not on any other coordinates $k \neq i,j$).

\begin{definition}[The $\Psi$ function]
\label{prop:psi}
For indices $i \neq j$ and $x_i, x_j \in (0,1)$, we have $$
\E[ A_i^{1/x_i - 1} A_j^{1/x_j - 1} ] = x_i x_j \Psi( x_i, x_j; \rho_i, \rho_j)
$$
for the function $\Psi$ defined as:
$$
\Psi(x_1,x_2; \rho_1,\rho_2)= \frac{ \E[ 
 A_1^{1/x_1-1} A_2^{1/x_2 - 1} ]  }{x_1 x_2 } \qquad \text{where $(A_1, A_2) \sim \text{DirichletCopula}(\rho_1, \rho_2)$}
$$
\end{definition}

\begin{proposition}[$\Psi$ monotonicity]
The function $\Psi(x_1, x_2; \rho_1, \rho_2)$ is nonincreasing in $\rho_1$ and $\rho_2$.
\end{proposition}
\begin{proof}
Let $\rho_1 < \rho_1'$. Take  $(A_1, A_2) \sim \text{DirichletCopula}(\rho_1, \rho_2)$ and $(A'_1, A_2') \sim \text{DirichletCopula}(\rho_1', \rho_2)$. By \Cref{prop:dirichlet-concordance}, we have $ (A_1', A_2') \preceq_{\mathrm c} (A_1, A_2) $ in the concordance order. Since the functions $z \mapsto z^{1/x_1 - 1}$  and $z \mapsto z^{1/x_2 - 1}$ are both nonnegative and increasing, integrating
the upper-orthant inequalities from \Cref{prop:dirichlet-concordance} gives
$$
\Psi(x_1, x_2; \rho_1', \rho_2) = \frac{ \E[ (A_1')^{1/x_1 - 1} (A_2')^{1/x_2-1}] }{x_1 x_2}  \leq \frac{ \E[ A_1^{1/x_1 - 1} A_2^{1/x_2-1}] }{x_1 x_2}  = \Psi(x_1, x_2; \rho_1, \rho_2)
$$

The monotonicity in $\rho_2$ is completely symmetric.
\end{proof}

The  $\Psi$ function is critical to analyzing the approximation factors  of the algorithms. Unfortunately, it is extremely complex, defined in terms of integrals and non-elementary functions. As a rough order-of-magnitude estimate, the following bound is useful: 
\begin{proposition}
\label{simp-corr}
There holds
$$
\Psi(x_1, x_2; \rho_1, \rho_2)  \leq  \binom{\rho_1 (1/x_1 - 1) + \rho_2 (1/x_2 - 1)}{\rho_1 (1/x_1 - 1)}^{-1}.$$
\end{proposition}

By contrast, note that \emph{independent} uniform random variables $U_1, U_2$ would have $$\E[U_1^{1/x_1 - 1} U_2^{1/x_2 - 1} ] = x_1 x_2.$$ Since $\binom{\rho_1 (1/x_1 - 1) + \rho_2 (1/x_2 - 1)}{\rho_1 (1/x_1 - 1)} > 1$,  \Cref{simp-corr} should be interpreted as a statement that the correlated uniform random variables $A_1$ and $A_2$ have strong negative correlation.

We show \Cref{simp-corr} in \Cref{sec: 03_Beta fcn}, along with other more-precise (but complicated) bounds on $\Psi$. For now, we calculate the algorithm behavior in terms of $\Psi$ viewed as a black-box function. We will later use the precise $\Psi$ bounds  to give explicit bounds for our algorithms.
\subsection{Bipartite rounding}\label{sec: 04_bipartite rounding}
The second part of the Dirichlet mechanism is to implement bipartite selection using the correlated uniform variables. We start with a bipartite graph $G = (U \cup V, E)$, with weight $x_e$ and a given, problem-specific correlation parameter $\rho_e$ for each edge $e = (u_e, v_e)$. Each left-node $u$ generates uniform random variables $A_e$ for the edges $e \in N(u)$ via \textsc{DirichletCopula}. These are then transformed into Exponential variables $Z_e$ with rate $x_e$. Each right-node $v$  selects the incident edge whose value $Z_e$ is the \emph{smallest.} (If $x(N(v)) < 1$ then there is a slight adjustment where $v$ may select no edges.)

If the Exponential random variables $Z_e$ were all \emph{independent}, then this would be equivalent to independent rounding; this is known as randomized rounding via exponential clocks in the literature~\cite{DBLP:journals/talg/AnNS17}. Intuitively, it is beneficial to negatively correlate the events that two adjacent edges are added, as this reduces the expected number of collisions and hence the need to subsequently remove edges to obtain a valid matching.

Combining the new sampling scheme with the algorithm from \cite{harris2024dependent}, we introduce \cref{alg:dependent-rounding}.

\begin{algorithm}[H]
\caption{\textsc{DepRound}($G, \vec{x},\vec{\rho}$): the Dirichlet rounding mechanism}
\begin{algorithmic}[1]
\REQUIRE $\vec x, \vec \rho \in [0,1]^E$ with $x(N(v)) \leq 1$ for all $v \in V$ and $\rho(N(u)) \leq 1$ for all $u \in U$.
\FOR{each left node $u\in U$ with edges $(u,v_1), \dots, (u,v_k)$}
\STATE \textcolor{gray}{//Sample $A_e$ via \textsc{DirichletCopula} and transform into Exponentials $Z_e$}
\STATE Draw $
(T_{(u,v_1)}, \dots, T_{(u,v_k)}) \sim \text{Dir}(\rho_{(u,v_1)}, \dots, \rho_{(u,v_k)}, 1 - \rho_{(u,v_1)} - \dots - \rho_{(u,v_k)}  )
$
\STATE Set $A_e \leftarrow I(T_e; \rho_e, 1 - \rho_e)$ and $Z_e = -\frac{\log A_e}{x_e}$ for each edge $e = (u,v_i) \in N(u)$.
\ENDFOR
    \FOR{each right-node \( v\in V \) and each edge $e \in N(v)$}
        \STATE \textcolor{gray}{//Choose an edge with (approximately) minimum $Z_f$ value to round to $1$.}
        \IF{edge $e$ satisfies the bound $
        (1 - x_e) Z_e < (x(N(v)) - x_e)  \min_{f \in N(v) \setminus \{e \}} Z_f$}
        \STATE Set $X_e = 1$         
        \ELSE 
        \STATE Set $X_e = 0$
        \ENDIF
    \ENDFOR
    \RETURN vector \(\vec{X}\)
\end{algorithmic}
\label{alg:dependent-rounding}
\end{algorithm}

Lines 3 and 4 are precisely \textsc{DirichletCopula} for the vector $(\rho_e: e \in N(u))$. Note that in Line 8, if $x(N(v))=1$, the condition reduces to choosing the minimum $Z_f$ value. To state the negative correlation results in greatest generality, we recall the following definition and lemma from \cite{harris2024dependent}.
\begin{definition}[Stable edge set]
An edge set $S \subseteq E$ of $G$ is \emph{stable} if it has no edges $e_1, e_2$ whose distance in the line graph of $G$ is precisely two. 
\end{definition}

\begin{lemma}[\cite{harris2024dependent}]
\label{prop:dep-round}
\cref{alg:dependent-rounding} satisfies the following properties:
\begin{enumerate}
        \item For any right-node $v \in V$, the random variables $Z_e: e \in N(v)$ are independent Exponentials, each of rate $x_e$.
        \item The random variables $A_e: e \in E$  are NA.
        \item For any edge $e \in E$, there holds $\Pr[X_e = 1] = x_e$.    
        \item Each right-node $v$ has at most one edge selected.  
        \item For a stable edge set $S$, there holds        
$\E\bigl[ \prod_{e \in S} X_e \bigr] \leq  \E \bigl[ \prod_{e \in S} A_e^{1/x_e - 1} \bigr] \leq \prod_{e \in S} x_e.
$
\end{enumerate}
    \end{lemma}
    \begin{proof}
The proofs are similar to \cite{harris2024dependent}. For completeness, we prove them fully here.
\begin{enumerate} 
    \item The random variables $T_e: e \in N(v) $ are generated independently, so $A_e: e \in N(v)$ are also independent. Since each $A_e$ is uniform, the inverse transform $Z_e = -\frac{\log A_e}{x_e}$ has Exponential distribution of rate $x_e$.
    \item Each variable $A_e$ is a monotone increasing function of $T_e$, so by \cref{thm: NA properties} the random variables $A_e: e\in E$ are NA. Edges corresponding to distinct nodes have independent values for $A_e$.
    \item The random variables $(Z_e)_{\{e \in N(v)\}}$ are independent Exponentials of rate $x_e$ respectively. By the well-known properties of exponential random variables, $Q := \frac{ x(N(v)) - x_e }{1 - x_e} \min_{f \in N(v) \setminus \{e \}} Z_f$ is an exponential random variable with rate $1 - x_e$. So $\Pr[X_e = 1] = \Pr [ Z_e <  Q ]  = \frac{x_e}{ x_e + (1 - x_e) } = x_e$.

    \item By hypothesis, we have $x(N(v)) \leq 1$. So if $X_e = 1$, then $Z_e < \frac{x(N(v)) - x_e}{1 - x_e} Z_f \leq Z_f$ for all other edges $f \in N(v) \setminus \{e \}$; clearly, there can be at most one edge with this property.
     \item 
   Each term $A_e^{1/x_e - 1}$ is an increasing nonnegative function of $A_e$. As the random variables $A_e$ are NA, by \cref{thm: NA properties}, we have $\E \bigl[ \prod_{e \in S} A_e^{1/x_e - 1} \bigr] \leq \prod_{e \in S} \E \bigl[ A_e^{1/x_e - 1} \bigr]$. Since each $A_e$ is a uniform random variable, it satisfies $\E[A_e^{1/x_e - 1}]  = x_e$.  This shows the bound $   \E \bigl[ \prod_{e \in S} A_e^{1/x_e - 1} \bigr] \leq \prod_{e \in S} x_e.
    $

    For the second bound, we may assume that all the right-nodes of the edges in $S$ are distinct, as otherwise $\prod_{e \in S} X_e = 0$ with probability one. Define $S'$ to be the set of edges outside $S$ that share a right-node with an edge in $S$, and let $W, W'$ denote the set of left-nodes of edge-sets $S$ and $S'$ respectively. We claim that $W, W'$ are disjoint. For, suppose that $u \in W \cap W'$. So there are edges
$(u, v) \in S, (u, v') \in S'$; by definition of $S'$, this implies that there is an edge $(u', v') \in S$. Then the edges $(u',v'), (u,v)$ have distance two in the line graph of $G$, contradicting that $S$ is stable.

  Now suppose we condition on all random variables corresponding to the nodes in $W$. In particular, this reveals all the random variables $A_S := ( A_{e} : e \in S )$. The random variables corresponding to nodes in $W'$ have their original unconditioned probability distributions.

    Consider an edge $e \in S$. Because we have assumed that edges in $S$ have distinct-right nodes, all edges $f \in N(v_e) \setminus \{e \}$ are in $S'$.  All such random variables correspond to nodes in $W'$ and retain their original, unconditioned probability distributions. In particular, the random variables $Z_{f}:f \in N(v_e) \setminus \{e\}$ are independent exponentials. So
 $Q := \frac{x(N(v_e)) - x_e}{1 - x_e} \min_{f \in N(v_e) \setminus \{e \}} Z_f$ is an exponential random variable with rate $1-x_e$,
 and hence
    \begin{align*}
        \E[ X_e \mid A_S ] & = \Pr [ Z_e < Q \; \big|\;A_{S} ]
= \mathrm e^{-(1-x_e) Z_{e}} = A_e^{1/x_e - 1}.
        \end{align*}
    
    Since the random variables $A_{S'}$ are NA, and each value $X_e$ is as a nonincreasing function of the distinct block of variables $A_f: f \in N(v_e) \setminus \{e \}$,  \cref{thm: NA properties} gives
    \begin{align*}
    \E \Bigl[ \prod_{e \in S} X_e \mid A_S \Bigr] \leq \prod_{e \in S} \E[ X_e \mid A_S ] = \prod_{e \in S} A_e^{1/x_e - 1}
    \end{align*}

    To finish, we apply iterated expectations with respect to random variables $A_S$.  \qedhere
    \end{enumerate}
\end{proof}

The strong negative correlation properties of the algorithm are stated entirely based on the $\Psi$ function.
    
\begin{theorem}\label{thm: E[U^q]}
    For any distinct edges $e_1 = (u,v_1)$, $e_2 =(u,v_2)$ with the same left-node $u$ we have
    $$
    \E[X_{e_1} X_{e_2}]\leq x_{e_1} x_{e_2} \Psi(x_{e_1}, x_{e_2}; \rho_{e_1}, \rho_{e_2}) 
    $$
\end{theorem}
\begin{proof} 
Note that $S = \{e_1, e_2\}$ is a stable edge-set, since $e_1$ and $e_2$ have distance one in the line graph of $G$. By \cref{prop:dep-round}, we have $
    \E[X_{e_1} X_{e_2} ] \leq \E \bigl[ A_{e_1}^{1/x_{e_1} - 1} A_{e_2}^{1/x_{e_2} - 1} \bigr].$ The random variables $A_{e_1}, A_{e_2}$ are generated by \textsc{DirichletCopula} with input $\vec{\rho}$. So the result follows from \cref{prop:psi}.
\end{proof}
\section{Online Matching algorithm}
\label{sec: 05_(online) matching}

Recall the problem setting: we have a bipartite graph $G = (U \cup V, E)$, where the offline nodes $U$ are fixed and known in advance. When each online node $v \in V$ appears, we learn the demands $g_{(u,v)} : u \in U$.  We are guaranteed that the weights $g_e$ form a fractional matching for the graph.

For edges $e = (u_e, v_e), f = (u_f, v_f)$, we say that $f < e$ if the online node $v_f$ is revealed before $v_e$.  For an edge $e = (u_e,v_e)$, we define $L(e) = \{ f \in N(u_e), f < e \}$, that is, edges which share a left-node with $e$ and come earlier in the ordering.

Our plan is to run $\textsc{DepRound}(G,\vec x,\vec \rho)$, where $\vec x$ and $\vec \rho$ are both determined in an online fashion from $\vec g$. For this, we take advantage of the fact that the generation of Dirichlet random variables, and the $\textsc{DepRound}$ algorithm, can both be implemented  online vertex-by-vertex. For full details see \cref{alg:online-rounding}.

\begin{algorithm}[H]
\caption{\textsc{ExponentialODRS}($\{g_e\in[0,1]: e\in E\}$)}
\begin{algorithmic}[1]
\STATE Define parameters $\alpha = 1.2337, \beta = 0.3$, and define function $F: [0,1] \rightarrow \mathbb R$ by
$$
F(t) = \frac{0.68145}{\sqrt{1 - 0.53562 t}}
$$

\STATE Initialize $M  = \emptyset$
\FOR{each online node $v\in V$ at its arrival}
    \FOR{each edge $e = (u,v) \in N(v)$}
    \STATE Calculate values
    $$
\hspace{-0.4in}    r_e = \sum_{f \in L(e)} g_f, \qquad y_e = \int_{t = r_e}^{r_e + g_e} F(t) \ \mathrm{d}t, \qquad  \rho_e =  \alpha \thinspace y_e, \qquad x_e = (1-\beta) F(r_e) g_e + \beta y_e
    $$
        \STATE Draw random variable $B_e \sim \text{Beta}(\rho_e, 1-\rho_e - \sum_{f \in L(e)} \rho_{f})$  and set
        $$
T_e = \Bigl( 1 - \sum_{f \in L(e)} T_{f} \Bigr) \cdot B_e, \qquad A_e = I(T_e, \rho_e, 1 - \rho_e), \qquad Z_e = -\frac{\log A_e}{x_e}
$$
\ENDFOR
\IF{there is any vertex $u \in U$ which is not matched in $M$ and there is an edge $e = (u,v) \in N(u)$ satisfying
$$
(1 - x_e) Z_e  <  (x(N(v)) - x_e)  \min_{f \in N(v)\setminus \{e \}} Z_f 
$$
}
\STATE Update $M \leftarrow M \cup \{e \}$ \COMMENT {Permanently commit edge $e$ to the matching}
\ENDIF
\ENDFOR 
\end{algorithmic}\label{alg:online-rounding}
\end{algorithm}

For intuition, suppose all values $g_e$ are infinitesimal. In this case, $x_e \approx F(r_e) g_e \approx y_e$ for each edge $e$.
The function $F$ serves as an ``attenuation factor" so that edge $e$ is selected with probability $x_e \approx F(r_e) g_e$, as opposed to $g_e$ in an offline algorithm.   Without this attenuation factor, the early arriving edges would go into the matching with a higher probability, as later ones can only be accepted if no earlier edge was selected. The factor $F(t)$ penalizes earlier edges (where $r_e \approx 0$ and $F(r_e) \approx F(0) = 0.68145$) compared to the final edges of a node (where $r_e \approx 1$ and $F(r_e) \approx 1$). This evens out the edges; as a side benefit, the attenuation factor increases the negative correlation strength of the Dirichlet mechanism.

To derive the specific functional form of $F$,  we assume the edge weights $g_e$ are infinitesimal, and solve a differential equation to maximize the minimum approximation ratio.  See \Cref{S-motivation-app} for further details.

Let us define
$$
Q(t_0,\Delta t) = \int_{t = t_0}^{t_0 + \Delta_t} F(t) \dt, \qquad \text{so that $y_e = Q(r_e, g_e)$.}
$$

By direct calculations, one can easily verify the following properties of the function $F$:
\begin{observation}\label{prop: numerical}
The function $F$ satisfies the following properties:
\begin{enumerate}
\item[(a)] $F$ is nonnegative, increasing, and concave-up;
\item[(b)] $Q(0,1) \leq 1/\alpha$;
\item[(c)] $F(1) \leq 1$.
\end{enumerate}
\end{observation}

We now examine a few basic properties of the algorithm. At a high level, we first show that the edge selection in \textsc{ExponentialODRS} is exactly the same as running \textsc{DepRound}, for suitable choice of parameters. This gives us \Cref{tre2}, which reduces the analysis of the algorithm to analyzing the $\Psi$ function.

\begin{proposition}
\label{tre4}
For an edge $e$ there holds $r_e = \sum_{f \in L(e)} g_f$ and $r_e + g_e \leq 1$ and $\sum_{f \in L(e)} y_f = Q(0,r_e) 
$
\end{proposition}
\begin{proof}
Let $u$ denote the left-node of $e$ and let the edges adjacent to $u$ be $f_1, f_2, \dots, f_k, e$ in order of arrival.  So $r_e = g_{f_1} + \dots + g_{f_k};$ since $g$ is a fractional matching this is at most $1 - g_e$.  We further compute:
    \begin{align*}
    &y_{f_1} + \dots + y_{f_k} = Q(r_{f_1},  g_{f_1}) + Q( r_{f_2}, g_{f_2}) + \dots + Q(r_{f_k}, g_{f_k})  \\
        & \qquad = Q(0, g_{f_1}) + Q(g_{f_1}, g_{f_2})  + \dots + Q(g_{f_1} + \dots + g_{f_{k-1}}, g_{f_k})  \\
        & \qquad = \int_{t=0}^{g_{f_1}} F(t) \dt + \int_{t=g_{f_1}}^{g_{f_1} + g_{f_2}} F(t) \dt  + \dots + \int_{t=g_{f_1} + \dots + g_{f_{k-1}}}^{g_{f_1} + \dots + g_{f_k}} F(t) \dt \\
        &\qquad = \int_{t=0}^{g_{f_1} + \dots + g_{f_k}} F(t) \dt = Q(0,r_e).    
        \qedhere
        \end{align*}
\end{proof}

\begin{proposition}
\label{tre1}
For each edge $e$ there holds $x_e \geq 0$ and $\rho_e \geq 0$. 

For each right-node $v$ there holds $\sum_{e \in N(v)} x_e \leq 1$. 

For each left-node $u$ there holds $\sum_{e \in N(u)} \rho_e \leq 1$.
\end{proposition}

\begin{proof}
Nonnegativity of $x$ and $\rho$ hold since $F$ is nonnegative. For each edge $e$, convexity and monotonicity of $F$ imply that $$
x_e = (1-\beta) g_e F(r_e) + \beta \int_{t=r_e}^{r_e+g_e} F(t) \dt \leq g_e F(1) \leq g_e.
$$

Therefore, for a right-node $v$, we have
$
        \sum_{e \in N(v)} x_e \leq \sum_{e \in N(v)} g_e  = g(N(v)) \leq 1
        $
        where the last inequality holds since $g$ is a fractional matching of $G$.

    For the second bound, consider adding an additional dummy edge $f$ at the end of $u$. By \Cref{tre4} and \Cref{prop: numerical}, we calculate $\sum_{e \in N(u)} \rho_e = \sum_{e \in L(f)} \alpha y_e = \alpha Q(0,r_f) \leq \alpha Q(0,1) \leq \alpha \cdot \frac{1}{\alpha} = 1.$
\end{proof}

Let us say now that an edge $e = (u,v)$ is \emph{selected} if $(1 - x_e) Z_e <  (x(N(v)) - x_e) \min_{f \in N(v)\setminus \{e \}} Z_f$, (i.e. it satisfies the inequality at Line 8); the edge may or may not be added to  the matching, depending on whether $u \in U$ has been matched earlier.

\begin{lemma}
\label{simlemma}
Let $X_e$ denote the indicator variable that edge $e$ is selected. The random variables $X_e: e \in E$ are as produced by  $\vec X \leftarrow \textsc{DepRound}(G, \vec x, \vec \rho)$.
\end{lemma}
\begin{proof}
The generation of $T$ in the algorithm follows from the well-known ``stick-breaking process" for the Dirichlet distribution \cite{gelman2003bayesian}. For completeness, we sketch a proof.

Let $u$  be a left-node with edges $e_1, e_2, \dots, e_k$ in order of arrival. By \Cref{tre1}, $\sum_{e \in N(u)} \rho_e \leq 1$. Let $T_{e_1}, \dots, T_{e_k}$ be random variables as defined in \textsc{ExponentialODRS} and let $(T'_{e_1}, \dots, T'_{e_k})$ be corresponding Dirichlet variables as in \textsc{DirichletCopula}.  We claim by induction that $$
(T_{e_1}, \dots, T_{e_i}) \sim \operatorname{Dir}(\rho_{e_1}, \dots, \rho_{e_i}, 1 - \rho_{e_1} - \dots - \rho_{e_i}).
$$

For, \Cref{dirichletfact}(4) shows that $\frac{ (T'_{e_i}| T'_{e_1}, \dots, T'_{e_{i-1}}) }{1 - \sum_{j=1}^{i} T'_{e_j}}$ has the same distribution as $\operatorname{Beta}(\rho_{e_i}, 1 - \sum_{j=1}^{i-1} \rho_{e_j})$, and $T_{e_i}$ is drawn according to this distribution. By \Cref{dirichletfact}(3) we aggregate the remaining coordinates $T'_{e_{i+1}}, \dots, T'_{e_k}$, hence $(T_{e_1}, \dots, T_{e_i})$ has precisely the same distribution as $(T'_{e_1}, \dots, T'_{e_i})$.

By \Cref{dirichletfact}(2), the overall distribution of \((T_e)_{e \in N(u)}\) is thus the same as in \textsc{DirichletCopula}.
\end{proof}

From \Cref{simlemma}, we get the following immediate consequences:
\begin{observation}
\label{tre2}
\begin{enumerate}
\item $M$ is a matching.
\item Each edge $e$ is selected with probability exactly $x_e$.
\item Edges $e,f$ sharing a left-node are jointly selected with probability at most $x_e x_f \Psi( x_e, x_f; \rho_e, \rho_f)$. 
\end{enumerate}
\end{observation}

The following calculation is now key to the approximation ratio:
\begin{lemma}\label{lem:correlation-factor}
For values $r_1, r_2, g_1, g_2 \geq 0$ with $g_2 + r_2 \leq r_1, g_1 + r_1 \leq 1$, and corresponding values $$
y_i = Q(r_i, g_i),\qquad x_i = (1 - \beta) F(r_i) g_i + \beta y_i,  \qquad \rho_i = \alpha y_i
$$
there holds
$$
 \Psi( x_1, x_2; \rho_1, \rho_2) x_1 x_2  \leq c \cdot F(r_1) g_1 y_2 + \frac{\beta (y_1 - g_1 F(r_1)) y_2  }{Q(0,r_1)} \qquad \qquad \text{for constant $c = 0.3939$}
 $$
\end{lemma}
Showing this requires a grid search over the relevant parameters $g_1, g_2, r_1, r_2$ and massive computer calculation, as discussed in Appendix~\ref{app-numerical}. It immediately gives us the bound on the approximation ratio:
\begin{theorem}\label{thm: main apx}
Each edge $e$ goes into the matching with probability at least $0.68 \thinspace g_e$.
\end{theorem}
\begin{proof}
We have $e \in M$ if $e$ is selected and there is no earlier edge $f \in L(e) $ which is selected. So,
$$
\Pr( e \in M) \geq  \E[X_e] - \sum_{f \in L(e)} \E[X_e X_f]
$$

Let $z = \beta (y_e - F(r_e) g_e)$ and let $s = Q(0,r_e)$. By \Cref{tre2} and \Cref{lem:correlation-factor}, we have
\begin{align*}
\E[X_e] &= x_e = F(r_e) g_e + z \\
\E[X_e X_f] &\leq  y_f ( c F(r_e) g_e + z / s)
\end{align*}

Thus, overall
\begin{align*}
\Pr( e \in M) &\geq F(r_e) g_e + z - \negthickspace \sum_{f \in L(e)} \negthickspace  y_f \bigl( c F(r_e) g_e + z / s \bigr) \\
&= z + g_e F(r_e) - s (c F(r_e) g_e + z/s) \tag{\Cref{tre4}}  \\
&= g_e F(r_e) (1 - c Q(0, r_e)) \tag{definition of $s$}
\end{align*}

It can be easily checked that $F(r) (1 - c Q(0,r)) \geq 0.68$ for all $r \in [0,1]$.
\end{proof}
\section{Weighted completion time on unrelated machines}\label{sec: 06_weighted completion time}

To reiterate the notation, we have a set of machines $\calM$ and a set of jobs $\calJ$. Each machine-job pair $(i,j)$ has a weight $w^{(i)}_j$ and processing time $p^{(i)}_j$. The objective is to schedule jobs to minimize $$
\text{OBJ} = \sum_i \sum_{\substack{\text{jobs $j$ assigned} \\ \text{to machine $i$}}}w_j^{(i)}C_j,
$$
where $C_j$ is the cumulative processing time on machine $i$ for all jobs assigned up to and including job $j$.

For a single machine, the problem can be solved with a well-known greedy heuristic \cite{smith1956}: schedule the jobs in a decreasing order of the \emph{Smith ratio} $\sigma_j=w_j/p_j$. Thus, the key task is to determine \emph{which} machine to allocate each task to.

Our algorithm starts with  a semidefinite-programming (SDP) relaxation, yielding fractional assignments $x_j^{(i)}$ for each $i,j$, as well as pairwise terms $x_{j,j'}^{(i)}$.  The precise form of the program is not relevant for us; see \cite{bansal2016lift} for details. Next, we use the clustering algorithm from \cite{harris2024dependent}: the jobs on each machine are partitioned into classes $\calP^{(i)}$ based on their processing times; to avoid worst-case behavior a random shift is used. Within each class $\calP^{(i)}$, the jobs are sorted in descending order of Smith ratios $\sigma_j^{(i)} = w_j^{(i)}/p^{(i)}_j$, and sequentially added to a cluster until the cumulative weight of the cluster, $\sum_{j\in\calC} x_j^{(i)}$, surpasses a threshold $\theta$. 

We obtain the integral solution $\vec X^{(i)}_j$ using the \textsc{DepRound} algorithm, where each edge  $e$ represents a possible assignment of a job $j$ to a cluster $\mathcal C^{(i)}_{k,\ell}$ (and hence to machine $i$).  Most of the rates $\rho_e$ are proportional to  $x^{(i)}_j$, however, the last job added to a cluster requires special handling when the cumulative weight of the cluster exceeds a larger threshold $\tau > \theta$. We refer to this as a \emph{truncated job} (there can be at most one per cluster), and the others as \textit{untruncated} jobs.

The full details are in \Cref{alg: scheduling}.

\begin{algorithm}[H]
    \caption{\textsc{Clustering} for the scheduling algorithm}
    \begin{algorithmic}[1]
    \STATE Define parameters $\pi=4.5, \theta = 0.56,\tau=0.608$.
    \STATE Draw a random variable $P_{\text {offset }}\sim \text{Unif}([0,1])$.
        \STATE
        Partition the jobs into processing time classes $\mathcal{P}_{k}^{(i)}=\bigl \{j: P_{\text {offset }}+\frac{\log p_{j}^{(i)}}{\log \pi} \in[k, k+1) \bigr \}$
        \FOR{each machine $i$ and class $\mathcal{P}_{k}^{(i)}$}
            \STATE
            Initialize cluster index $\ell\gets1$ and set $\mathcal{C}_{k, 1}^{(i)}=\emptyset$.
            \STATE Sort $\mathcal{P}_{k}^{(i)}$ in decreasing order of Smith ratio as $j_{1}, j_{2}, \ldots,  j_s$.
            \FOR{$t=1, \ldots, s$ and each job $j=j_{t}$}
          \STATE Update $C^{(i)}_{k,\ell} \leftarrow C^{(i)}_{k,\ell} \cup \{j \}$.
                \IF{$x^{(i)} ( \mathcal{C}_{k, \ell}^{(i)} ) > \tau$}
                \STATE Set $\mathcal T^{(i)}_{k,\ell} = \{ j \}$ and $\mathcal U^{(i)}_{k,\ell} = \mathcal C^{(i)}_{k,\ell} \setminus \{j \}$
                \STATE Set $\rho_{j'}^{(i)} =  x_{j'}^{(i)}/\tau$ for $j' \in \mathcal{U}_{k, \ell}^{(i)}$ and set $\rho_{j}^{(i)} = 1 - x^{(i)} ( \mathcal{U}_{k, \ell}^{(i)} ) / \tau$.
                \STATE Update $\ell \leftarrow \ell+1$ and initialize the new cluster $\mathcal{C}_{k, \ell}^{(i)}=\emptyset$.                                
                \ELSIF{$x^{(i)}( \mathcal{C}_{k,\ell}^{(i)})  \geq \theta$ or $t = s$}        
                \STATE Set $\mathcal T^{(i)}_{k,\ell} = \emptyset$ and $\mathcal U^{(i)}_{k,\ell} =  \mathcal C^{(i)}_{k,\ell}$.                
            \STATE Set $\rho_{j'}^{(i)} =  x_{j'}^{(i)} /x^{(i)}( \mathcal{C}_{k,\ell}^{(i)})$ for each $j' \in \mathcal{C}_{k, \ell}^{(i)}$ 
                              \STATE Update $\ell \leftarrow \ell+1$ and initialize the new cluster $\mathcal{C}_{k, \ell}^{(i)}=\emptyset$.
                \ENDIF                
            \ENDFOR
    \ENDFOR
    \STATE Form a bipartite graph $G$: there is a left-node for each cluster $\mathcal C^{(i)}_{k,\ell}$ and a right node for each job $j$. There is an edge $e = ( C^{(i)}_{k,\ell}, j)$ for each $j \in C^{(i)}_{k,\ell}$, with $\rho_e = \rho^{(i)}_j$ and $x_e = x^{(i)}_j$.

    \STATE Run algorithm $\textsc{DepRound} (G, \vec x, \vec \rho)$ to convert the fractional solution $\vec x$ into an integral solution $\vec X$.
    \end{algorithmic}\label{alg: scheduling}
\end{algorithm}

The subsequent discussion relies on the following key fact.
\begin{theorem}\cite{bansal2016lift}\label{thm: approx eta} For any machine $i^{*}$, suppose the jobs $1, \ldots, n$ are sorted in descending order of the Smith ratio $\sigma_{j}^{\left(i^{*}\right)}$. For each job $j^{*}$,  define
$$
\begin{aligned}
Z^{\left(i^{*}, j^{*}\right)} & =\frac{1}{2}\left(\sum_{j=1}^{j^{*}} X_{j}^{\left(i^{*}\right)}\left(p_{j}^{\left(i^{*}\right)}\right)^{2}+\sum_{j=1}^{j^{*}} \sum_{j^{\prime}=1}^{j^{*}} X_{j}^{\left(i^{*}\right)} X_{j^{\prime}}^{\left(i^{*}\right)} p_{j}^{\left(i^{*}\right)} p_{j^{\prime}}^{\left(i^{*}\right)}\right) \\
\mathrm{LB}^{\left(i^{*}, j^{*}\right)} & =\frac{1}{2}\left(\sum_{j=1}^{j^{*}} x_{j}^{\left(i^{*}\right)}\left(p_{j}^{\left(i^{*}\right)}\right)^{2}+\sum_{j=1}^{j^{*}} \sum_{j^{\prime}=1}^{j^{*}} x_{j, j^{\prime}}^{\left(i^{*}\right)} p_{j}^{\left(i^{*}\right)} p_{j^{\prime}}^{\left(i^{*}\right)}\right)
\end{aligned}
$$

    If $\E[Z^{\left(i^{*}, j^{*}\right)}] \leq \eta \cdot \mathrm{LB}^{\left(i^{*}, j^{*}\right)}$ for all $i^{*}, j^{*}$, then the schedule is an $\eta$-approximation in expectation.
\end{theorem}

In light of \cref{thm: approx eta}, we focus on a single pair $(i^*,j^*)$. From now on we omit the superscript $(i^*)$ and $(j^*)$ from all relevant variables, writing $Z$, $X_j$, $p_j$, $w_j$, $\sigma_j$, $\calP_k$,
instead of $Z^{(i^*,j^*)}$, $X_j^{(i^*)}$, $p_j^{(i^*)}$,  etc.. In addition, all summations $\sum_j$ are assumed to range over $j\in\mathcal J^*$, the set of jobs whose Smith ratio is at least that of  the target job $j^*$. Similarly, summations \(\sum_{j', j''}\) range over \textit{ordered} pairs $(j',j'')\in \mathcal{J}^*\times\mathcal{J}^*$.

For each processing time class \(\mathcal{P}_k\), we define $P_k = \pi^{k-P_{\text{offset}}}, \mathcal{P}_k^* = \mathcal{P}_k \cap \mathcal{J}^*$, and similarly for each cluster \(\mathcal{C}_{k, \ell}\), we set \(\mathcal{C}_{k, \ell}^* = \mathcal{C}_{k, \ell} \cap \mathcal{J}^*\). The last cluster opened within $\calP_k^*$ is called a \emph{leftover cluster} (it could be empty) and is denoted by \(\mathcal{C}_{k, \text{left}}^*\). For a job $j\in\mathcal{P}_k^{(i)}$, we define $H_j = p_j/P_k \in [1,\pi]$.

\begin{proposition}
\label{xjjprop}
For any pair of distinct jobs $j, j'$, we have $\E[ X_j X_{j'} ] \leq x_j x_{j'}$. Moreover, if $j, j'$ are both in the same cluster, then $\E[ X_j X_{j'} ] \leq x_j x_{j'} \Psi(x_j, x_{j'}; \rho_j, \rho_{j'})$.
\end{proposition}
\begin{proof}
If jobs $j, j'$ are in clusters $\mathcal C^{(i^*)}_{k,\ell}, \mathcal C^{(i^*)}_{k',\ell'}$ for machine $i^*$, then $\{ ( (i^*,k,\ell), j), ( (i^*, k', \ell'), j') \}$ is a stable edge set of $G$, since every job is assigned to a unique cluster per machine. So the first bound follows from \Cref{prop:dep-round}.  If jobs $j, j'$ are in the same cluster, then the edges $( (i^*,k,\ell), j), ( (i^*, k, \ell), j')$ share the same left-node of $G$. So the second bound follows from \Cref{thm: E[U^q]}.
\end{proof}

Now the objective becomes to bound the upper-bound and lower-bound variables
\begin{align*}
Z &=\frac{1}{2} \Bigl (\sum_{j} X_{j} p_{j}^{2}+\sum_{j, j^{\prime}} X_{j} X_{j^{\prime}} p_{j} p_{j^{\prime}}\Bigr)=\sum_{j} X_{j} p_{j}^{2}+\frac{1}{2} \sum_{j, j^{\prime}: j^{\prime} \neq j} X_{j} X_{j^{\prime}} p_{j} p_{j^{\prime}} \\
\mathrm{LB} &=\frac{1}{2}\Bigl(\sum_{j} x_{j} p_{j}^{2}+\sum_{j, j^{\prime}} x_{j, j^{\prime}} p_{j} p_{j^{\prime}}\Bigr)  \quad =\sum_{j} x_{j} p_{j}^{2}+\frac{1}{2} \sum_{j, j^{\prime}: j^{\prime} \neq j} x_{j, j^{\prime}} p_{j} p_{j^{\prime}}.
\end{align*}

We emphasize that $Z$ is a random variable which is determined by two sources of randomness: the offset $P_{\text{offset}}$ used to form the clusters, and the Dirichlet mechanism itself. Conversely, $\mathrm{LB}$ is just a scalar quantity.  In order to analyze the lower and upper bounds, we define parameters
$$
L = \sum_{j}x_jp_j, \qquad \qquad
Q = \sum_{j}x_jp_j^2,
$$

Note that $\mathrm{LB} \geq \max \left\{Q, \frac{1}{2}\left(Q+L^{2}\right)\right\}$. On the other hand, just using the negative correlation property $E[ X_j X_{j'} ] \leq x_j x_{j'}$, we have $\E[Z]\leq Q+\frac{1}{2}L^2$. So $\E[Z]\leq 3/2\cdot\mathrm{LB}$, hence, providing a $1.5$-approximation. \\
 
 Our argument closely follows \cite{harris2024dependent}, and only a small portion depends on the negative correlation itself. In particular, the strong negative correlation property, which is where the Dirichlet mechanism plays a role, leads to an upper bound on $\E[Z]$ expressed in terms of $L, Q$ and some additional parameters. A separate argument provides a lower bound on $\mathrm{LB}$ in terms of the same parameters. Finally, a numerical optimization determines the maximum ratio between these bounds, which is the upper bound on the approximation ratio. 
 
 \begin{theoremstar}\label{thm: sch_apx}
We have
\[
\E[Z] \le 1.387 \cdot \mathrm{LB}.
\]
In particular, \Cref{alg: scheduling} achieves an approximation ratio of $1.387$ in expectation, and the SDP relaxation has integrality gap at most $1.387$.
\end{theoremstar}

The technical details are deferred to \Cref{ap: sch}.

\section{Analysis of the function $\Psi$}\label{sec: 03_Beta fcn}

We consider a number of properties of the $\Psi$ function, via analysis of the Incomplete Beta function. This will allow us to compute various limits as well as  arbitrarily-accurate approximations over parameter ranges. 

Throughout this section, we define the key functions
\[
        I_\rho(t)=I(t;\rho,1-\rho), \qquad          \phi_\rho(t)
        =
        \frac{\rho B(t;\rho,1-\rho)}{t^\rho}
\]
related by the Beta function identity
\begin{equation}
\label{betaident}
    I_\rho(t)
        =
        \frac{t^\rho\phi_\rho(t)}{ \Gamma(1+\rho)\Gamma(1-\rho)}.
        \end{equation}
        
We begin by recording a few facts about the power series expansion of the incomplete Beta function and related functions.

\begin{lemma}[Series expansion of $\phi$]
\label{lem:seriesphi}
Let $\rho \in (0,1)$. The function $\phi_\rho$ extends to an analytic, zero-free function on
the open unit disk, with $\operatorname{Re}\phi_\rho(t)>0$ for $|t| < 1$.  There is an analytic branch of $\log\phi_\rho$ on the open unit disk,
normalized by $\log\phi_\rho(0)=0$, with
\[
        \log\phi_\rho(t)
        =
        \sum_{k\geq1}c_k(\rho)t^k, \qquad \text{with $c_k(\rho) > 0$ for $k \geq 1$}
\]

Moreover, for any $q \geq 0$, we have the series expansion
$$
\phi_{\rho}(t)^q = \sum_{k \geq 0} b_k(\rho, q) t^k, \qquad \text{with $b_k(\rho, q) \geq 0$ for all $k$},
$$
and $b_k$ values are given explicitly by the recurrence
$$
b_0(\rho,q)=1, \qquad   b_k(\rho,q)
=
\frac1k
\sum_{j=1}^k
\bigl((q+1)j-k\bigr)
\frac{\rho (\rho+j-1)!}{(\rho+j)(\rho-1)! j!} b_{k-j}(\rho,q) \quad \text{for $k \geq 1$}
$$
\end{lemma}
\begin{proof}
It is a well-known fact about the Beta function that $\phi$ has an analytic expansion
\begin{equation}
\label{phi-expansion}
\phi_\rho(t)
=
\sum_{k\geq0} \frac{\rho (\rho+k - 1)!}{(\rho+k) (\rho-1)! k!}t^k.
\end{equation}
which converges on the open unit disk with $\operatorname{Re}\phi_\rho(t)>0.$ Thus $\phi_\rho$ has no zeros in the open unit disk, and there is an
analytic branch of $\log\phi_\rho$ on this disk with
$\log\phi_\rho(0)=0$.

In \Cref{sec: ap1}, specifically \Cref{cor: ck>0}, we show a criterion for the logarithm of a given function to have positive derivatives. Here, we apply it specifically to the function $\phi_{\rho}$.  Here, from our expansion (\ref{phi-expansion}), we get
$$
a_k(\rho) = \phi_{\rho}^{(k)}(0)/k! = \frac{\rho (\rho+k - 1)!}{(\rho-1)! k! \rho+k}.
$$

Thus $a_0 = 1$, and $a_k(\rho)>0$ for every \(k\geq 0\).  Moreover,
\[
\frac{a_{k+1}(\rho)}{a_k(\rho)}
   =\frac{(\rho+k)^2}{(k+1)(\rho+k+1)}.
\]
which is increasing in $k$. Hence the sequence
\(\{a_k(\rho)\}_{k\geq0}\) is log-convex. Hence, the normalized derivatives $c_k$ of $\log \phi$ at zero are positive for $k \geq 1$. So $\phi_{\rho}$ is log-absolutely monotonic. 

We use \Cref{cor:monpower} to obtain the explicit recurrence formula for the derivatives of $\phi_{\rho}(t)^q$ in terms of the normalized derivatives $a_k(\rho)$ of $\phi_{\rho}$. These are nonnegative since $\phi_{\rho}$ is log-absolutely monotonic.  From analyticity, $\phi_{\rho}(t)^q$ is equal to its expansion $\sum_j b_j(\rho, q) t^j$.
\end{proof}
\begin{lemma}
\label{mainestlemma0}
Suppose that \(x_1,x_2\in(0,1]\) and
\(\rho_1,\rho_2\in(0,1)\), with \(\rho_1+\rho_2\leq1\).  For
\(i=1,2\), define
\[
\lambda_i=\rho_i/x_i,
\qquad
q_i=1/x_i - 1,
\qquad
\theta_i=
\bigl(\Gamma(1+\rho_i)\Gamma(1-\rho_i)\bigr)^{-q_i}.
\]
For $j \geq 0$ define
$$
\kappa_{i,j} = \theta_i \frac{\lambda_i}{\lambda_i + j} \binom{\lambda_i + j}{\rho_i} b_j(\rho_i, q_i)
$$
where $b_j(\rho, q)$ are the coefficients in \Cref{lem:seriesphi}.

For \(j_1,j_2\geq0\), define
\begin{align*}
\alpha_{j_1,j_2}
&=
\binom{\rho_1q_1+\rho_2q_2+j_1+j_2}
      {\rho_1q_1+j_1}^{-1} =
\binom{\lambda_1(1-x_1)+\lambda_2(1-x_2)+j_1+j_2}
      {\lambda_1(1-x_1)+j_1}^{-1}.
\end{align*}
Then the $\kappa_{i,j}$ values are nonnegative and satisfy $\sum_{j} \kappa_{i,j} = 1$ and 
\[
\Psi(x_1,x_2;\rho_1,\rho_2)
=
\sum_{j_1,j_2\geq0}
\alpha_{j_1,j_2}\kappa_{1,j_1}\kappa_{2,j_2}.
\]
\end{lemma}

\begin{proof}
Fix \(i\), and temporarily suppress the subscript \(i\).  From \Cref{lem:seriesphi} we have  $\phi_{\rho}(t)^q = \sum_{j \geq 0} b_j t^j$, 
where we write $b_j := b_j(\rho, q)$ for brevity. With our identity (\ref{betaident}), this gives
\begin{equation}
\label{yah0a}
I_{\rho}(t)^q=\theta \sum_{j\geq0}b_j t^{\rho q + j}.
\end{equation}

Now note that, for $Z \sim \text{Beta}(\rho, 1- \rho)$, we have 
\begin{align*}
\mathbb{E}\bigl[Z^{\rho q+j}\bigr] =  \int_0^1t^{\rho q+j}\,\dd I_\rho(t) = 
\frac{\lambda x}{\lambda+j}
\binom{\lambda+j}{\rho}
\end{align*}
and hence
\begin{equation}
\kappa_j
=
\frac{\theta}{x}b_j
 \E\bigl[Z^{\rho q+j}\bigr] =  \frac{\theta_i}{x_i}b_{j} \int_0^1t^{\rho q+j}\,\dd I_\rho(t).
\label{yah1}
\end{equation}
This is nonnegative since $b_j \geq 0$.

Since \(I_\rho\) is the continuous distribution function of \(Z\),
the random variable \(I_\rho(Z)\) is uniform on \([0,1]\). From  (\ref{yah0a}),  and Tonelli's theorem, we have 
\[
\frac{1}{q+1} = \E_{U \sim \text{Unif}[0,1]}[U^q] = \E[I_\rho(Z)^q\bigr]
=
\theta\sum_{j\geq0}b_j
\mathbb{E}\bigl[Z^{\rho q+j}\bigr].
\]
Therefore, by (\ref{yah1}), 
\[
\sum_{j\geq0}\kappa_j = \sum_j \frac{\theta}{x} b_j \E[ Z^{\rho q + j}] 
=
\frac1{x(q+1)}
=
1.
\]
This proves the asserted properties of \(\kappa_{i,j}\) for each
\(i=1,2\).

It remains to derive the formula for \(\Psi\).  Let $(X_1,X_2)
\sim
\operatorname{Dir}
\bigl(\rho_1,\rho_2,1-\rho_1-\rho_2\bigr)$. 
The marginal distribution of \(X_i\) is
\(\operatorname{Beta}(\rho_i,1-\rho_i)\), so (\ref{yah1}) above gives
\[
\kappa_{i,j}
=
\frac{\theta_i}{x_i}b_{i,j}
\mathbb{E}\bigl[X_i^{\rho_i q_i + j}\bigr].
\]

The Dirichlet moment formula \Cref{fact:prop-of-dirichlet}(6) gives
$$
\frac{\mathbb{E}[X_1^{\rho_1 q_1 + j_1}X_2^{\rho_2 q_2 + j_2}]}
     {\mathbb{E}[X_1^{\rho_1 q_1 + j_1}]
      \mathbb{E}[X_2^{\rho_2 q_2 + j_2}]} =
\binom{\rho_1 q_1+\rho_2 q_2 + j_1 + j_2}{\rho_1 q_1 + j_1}^{-1} =\alpha_{j_1,j_2}.
$$

Using the definition of \(\Psi\) and the expansion in (\ref{yah0a}),  we obtain
\begin{align*}
\Psi(x_1,x_2;\rho_1,\rho_2)
&=
\frac{1}{x_1x_2}
\mathbb{E}\left[
 I_{\rho_1}(X_1)^{q_1}
 I_{\rho_2}(X_2)^{q_2}
\right] \\
&=
\frac{\theta_1\theta_2}{x_1x_2}
\sum_{j_1,j_2\geq0}
b_{1,j_1}b_{2,j_2}
\mathbb{E}\left[
 X_1^{\rho_1 q_1 + j_1}X_2^{\rho_2 q_2 + j_2}
\right] \\
&=
\sum_{j_1,j_2\geq0}
\alpha_{j_1,j_2}
\left(
 \frac{\theta_1}{x_1}b_{1,j_1}
 \mathbb{E}[X_1^{\rho_1 q_1 + j_1}]
\right)
\left(
 \frac{\theta_2}{x_2}b_{2,j_2}
 \mathbb{E}[X_2^{\rho_2 q_2 + j_2}]
\right) \\
&=
\sum_{j_1,j_2\geq0}
\alpha_{j_1,j_2}\kappa_{1,j_1}\kappa_{2,j_2}.
\qedhere
\end{align*}
\end{proof}

\begin{lemma}
\label{mainestlemma}
With the notation of \Cref{mainestlemma0}, for any integers $k_1, k_2 \geq 0$ we have 
\begin{align*}
&\Psi(x_1, x_2; \rho_1, \rho_2) \leq \sum_{\substack{j_1<k_1 \\ j_2<k_2}} \alpha_{j_1,j_2}\kappa_{1,j_1} \kappa_{2,j_2} + \sigma_{2,k_2} \sum_{j_1<k_1} \alpha_{j_1,k_2}\kappa_{1,j_1}
    + \sigma_{1,k_1} \sum_{j_2<k_2} \alpha_{k_1,j_2}\kappa_{2,j_2}  +  \sigma_{1,k_1} \sigma_{2,k_2} \alpha_{k_1,k_2}
\end{align*}
where $\sigma_{i,k} = 1 - \sum_{j < k} \kappa_{i,j}$.
\end{lemma}
\begin{proof}
Note that $\alpha_{j_1,j_2}$ is a \emph{decreasing} function of $j_1$ as $$
\alpha_{j_1+1,j_2} / \alpha_{j_1,j_2} = \frac{\rho_1 q_1 + j_1 + 1}{\rho_1 q_1 + \rho_2 q_2 + j_1 + j_2 + 1} \leq 1;$$
by a completely symmetric argument it is also a decreasing function of $j_2$.
 
     Using \Cref{mainestlemma0}, nonnegativity of $\kappa$, and monotonicity of $\alpha_{j_1,j_2}$, we thus get
    \begin{align*}
    &\Psi(x_1,x_2;\rho_1,\rho_2) = \sum_{j_1\ge 0, j_2\ge 0}\alpha_{j_1,j_2}\kappa_{1,j_1} \kappa_{2,j_2} \\
    & \quad \leq  \sum_{\substack{j_1<k_1 \\ j_2<k_2}} \alpha_{j_1,j_2}\kappa_{1,j_1} \kappa_{2,j_2} +\sum_{\substack{j_1<k_1 \\ j_2\geq k_2}}\alpha_{j_1,k_2}\kappa_{1,j_1} \kappa_{2,j_2}    +\sum_{\substack{j_1\geq k_1 \\ j_2<k_2}}\alpha_{k_1,j_2} \kappa_{1,j_1} \kappa_{2,j_2}     +\sum_{\substack{j_1\geq k_1 \\ j_2\geq k_2}} \alpha_{k_1,k_2}\kappa_{1,j_1} \kappa_{2,j_2} \\
& \quad =\sum_{\substack{j_1<k_1 \\ j_2<k_2}} \alpha_{j_1,j_2}\kappa_{1,j_1} \kappa_{2,j_2} + \sum_{j_1<k_1} \alpha_{j_1,k_2}\kappa_{1,j_1} \Bigl(1-\sum_{j_2<k_2}
    \kappa_{2,j_2}\Bigr)
    + \sum_{j_2<k_2} \alpha_{k_1,j_2}\kappa_{2,j_2}\Bigl(
  1  -\sum_{j_1<k_1}\kappa_{1,j_1}
    \Bigr)\\
    & \qquad \qquad + \alpha_{k_1,k_2}\Bigl(
    1 -  \sum_{j_2<k_2}\kappa_{2,j_2} -\sum_{j_1<k_1}\kappa_{1,j_1} + 
    \sum_{\substack{j_1< k_1 \\ j_2< k_2}} \kappa_{1,j_1} \kappa_{2,j_2}
    \Bigr)
    \end{align*}
    where the final line uses $\sum_j \kappa_{i,j} = 1$. With some rearrangement of terms, this gives the claimed bound.
    \end{proof}

We refer to \Cref{mainestlemma} as the $(k_1, k_2)$-order approximation of $\Psi$. If $k_1 = k_2 = k$, we refer to it simply as the $k^{\text{th}}$-order approximation.  \Cref{simp-corr} is simply the $0^{\text{th}}$-order approximation. 

\begin{corollary}
\label{poisson-prop}
Let $\lambda_1, \lambda_2 > 0$ be fixed positive reals.  For $x_1, x_2 > 0$, there holds
$$
\lim_{x_1, x_2 \rightarrow 0} \Psi( x_1, x_2; \lambda_1 x_1, \lambda_2 x_2) = \binom{\lambda_1 + \lambda_2}{\lambda_1}^{-1}
$$
\end{corollary}
\begin{proof}
For an upper bound on the limit, apply \Cref{mainestlemma} with $k_i = 0, \rho_i = \lambda_i x_i$; note that $\alpha_{0,0} \rightarrow \binom{\lambda_1 + \lambda_2}{\lambda_1}^{-1}$ as $x_1, x_2 \rightarrow 0$. For the lower bound, \Cref{mainestlemma0} and the fact that the terms $\kappa, \alpha$ are nonnegative gives $\Psi( x_1, x_2; \rho_1, \rho_2) \geq \alpha_{0,0} \kappa_{1,0} \kappa_{2,0}$. It can be checked by routine calculations that $\kappa_{1,0}, \kappa_{2,0} \rightarrow 1$ in the limit.
\end{proof}

\subsection{Numerical bounds on $\Psi$}
For algorithmic applications, we typically need to bound $\Psi$ within a given parameter range. The pointwise estimates, e.g. \Cref{mainestlemma}, are not sufficient for this task. We here discuss a few relevant issues on how to bound $\Psi$ tightly and accurately.

\begin{proposition}
\label{gratta1}
For integers $k_1, k_2 \geq 0$, suppose we are given upper-bound values \begin{eqnarray*}
&\kappa_{i,j}^+ \geq \kappa_{i,j} &  \quad (i=1,2,\ 0\le j<k_i) \\
&\alpha_{j_1, j_2}^+ \geq \alpha_{j_1,j_2} & \quad (0\le j_1\le k_1,\ 0\le j_2\le k_2).
\end{eqnarray*}

For $i = 1,2$ and $j = 0, 1, \dots, k_i$ define values
$$
 S_{i,j} = \min\{1, \kappa_{i,0}^+ + \dots, \kappa_{i,j-1 }^+ \}
$$
and define corresponding values $\hat \kappa$ by
\begin{align*}
\hat \kappa_{i,j} &= S_{i,j+1} - S_{i,j}, \qquad \text{for $j = 0, \dots, k_i-1 $} \\
\hat \kappa_{i,k_i} &= 1 - S_{i,k_i}
\end{align*}

Then 
\begin{align*}\Psi(x_1, x_2; \rho_1, \rho_2) &\leq \sum_{\substack{j_1 \leq k_1 \\ j_2 \leq k_2}} \alpha_{j_1,j_2}^+ \hat \kappa_{1,j_1} \hat \kappa_{2,j_2}
\end{align*}
\end{proposition}
\begin{proof}
For $i=1,2$, define vectors
$p_i,w_i$ indexed by $\{0,\ldots,k_i\}$ by
\[
p_i(j)=
\begin{cases}
\kappa_{i,j},&j<k_i, \\
\sigma_{i,j}  &j = k_i
\end{cases}
\qquad
w_i(j)=
\begin{cases}
\hat \kappa_{i,j},&j<k_i,\\
1 - S_{i,k_i} &j=k_i.
\end{cases}
\]
By \Cref{mainestlemma0}, the vectors $p_i$ are nonnegative, while the vectors $w_i$ are nonnegative by definition of $S$.

Observe that for every $t<k_i$, we have
\[
\sum_{j=0}^t p_i(j)
=
\sum_{j=0}^t\kappa_{i,j}
\le
\min \Bigl \{1, \sum_{j=0}^t \kappa^+_{i,j} \Bigr \}
=
S_{i,t+1} = 
\sum_{j=0}^t \hat \kappa_{i,j}.
\]
while for $t=k_i$, we have
\[
\sum_{j=0}^{k_i}p_i(j)=1
=
\sum_{j=0}^{k_i}w_i(j).
\]
Hence $w_i$ prefix-sum dominates $p_i$.

\Cref{mainestlemma} can be written as
\begin{equation}
\label{bit0}
\Psi(x_1,x_2;\rho_1,\rho_2)
\le
\sum_{j_1=0}^{k_1}\sum_{j_2=0}^{k_2}
\alpha_{j_1,j_2}p_1(j_1)p_2(j_2).
\end{equation}
Since $w_1$ prefix-sum dominates $p_1$ and $w_2$ prefix-dominates $p_2$, while the values $\alpha_{j_1, j_2}$ are nonnegative and nonincreasing in both coordinates,  we have
\begin{equation}
\sum_{j_1,j_2}
\alpha_{j_1,j_2}p_1(j_1)p_2(j_2)
\le
\sum_{j_1,j_2}
\alpha_{j_1,j_2}w_1(j_1)w_2(j_2).
\end{equation}

Since $\alpha^+_{j_1,j_2}\ge\alpha_{j_1,j_2}$ and all
entries of $w_1,w_2$ are nonnegative,
\begin{equation}
\sum_{j_1,j_2}
\alpha_{j_1,j_2}w_1(j_1)w_2(j_2)
\le
\sum_{j_1,j_2}
\alpha_{j_1,j_2}^+ w_1(j_1)w_2(j_2).
\label{bit3}
\end{equation}
Combining Eqs.~\labelcref{bit0}--\labelcref{bit3} and expanding the final double sum
yields the claimed bound.
\end{proof}

\begin{proposition}
\label{psicompbound}
Let $\lambda_i^-$ and $x_i^+ \lambda_i^- < 1$ for $i = 1,2$. Suppose we are given a box $$
B = [x_1^-, x_1^+] \times [x_2^-, x_2^+] \times [\lambda_1^-, \infty) \times [\lambda_2^-, \infty)
$$
and we want to upper-bound $\Psi$ over $(x_1, x_2, \lambda_1, \lambda_2) \in B$, i.e. where $x_i \in [x_i^-, x_i+], \rho_i/x_i \in [\lambda_i^-, \infty)$. 

We first use the $\rho$-monotonicity bound 
$$
\Psi(x_1, x_2; \rho_1, \rho_2) \leq \Psi(x_1, x_2; \lambda_1^{-} x_1, \lambda_2^{-} x_2)
$$

Next, after fixing $\rho_i = \lambda_i^- x_i$, we can further bound the terms as 
\begin{align*}
\alpha_{j_1, j_2} &\leq \binom{\lambda_1^{-} (1 - x_1^+) + \lambda_2^{-} (1 - x_2^+) + j_1 + j_2}{\lambda_1^{-1} (1 - x_1^+) + j_1}^{-1} \\
\theta_i &\leq \theta_i^+ := \begin{cases}
1 & \text{if $x_i^- = 0$} \\
(\Gamma(1+x_i^{-} \lambda_i^{-}) \Gamma(1-x_i^{-} \lambda_i^{-}))^{1-1/x_i^-} & \text{if $x_i^- > 0, x_i^+ \leq 1/2$} \\
(\Gamma(1+x_i^{-} \lambda_i^{-}) \Gamma(1-x_i^{-} \lambda_i^{-}))^{1-1/x_i^+} & \text{if $x_i^- > 0, x_i^+ > 1/2$}
\end{cases} \\
\kappa_{i,j}  &\leq \theta_i^+ \cdot \frac{\lambda_i^-}{\lambda_i^- + j} \cdot \binom{\lambda_i^-+j}{\lambda^-_i x^*_{i,j}} \cdot \max_{z \in [x_i^-, x_i^+]} b_j(\lambda_i^- z, 1/z - 1) \\
& \qquad \quad \text{for } x^*_{i,j} := \text{Clamp} \Bigl( \frac{\lambda_i^- + j}{2 \lambda_i^-} ; [x_i^-, x_i^+] \Bigr)
\end{align*}

Observe that $b_j$ is a rational function for every fixed $j$; we can maximize it over the interval using standard algebra algorithms.
\end{proposition}

Using such bounds, we can use \Cref{gratta1} to obtain an overall upper bound on $\Psi$. Note that we always have the generic bound $\Psi \leq 1$, which is useful in some edge cases.

\section{Acknowledgments}
We thank Mark Jacobson for suggesting the use of a Dirichlet distribution, and for explaining its definitions and properties.

\appendix

\section{Criterion for log-absolutely monotonic functions}
\label{sec: ap1}

We need to show that a function related to the Incomplete Beta function is log-absolutely monotonic. We will do this by deriving a general criterion for functions with fast-growing derivatives. This may be of independent interest. 

\begin{fact}[Fa\'{a} di Bruno formula]
\label{faa}

The $j^{\text{th}}$ derivative of a composite function \( \xi(\eta(x)) \) is given by:
\[
\frac{{\mathrm d} ^j}{{\mathrm d} x^j} \xi(\eta(x)) = \sum_{\substack{ k_1, \dots, k_j \geq 0: \\ k_1 + 2 k_2 + \dots + j k_j = j}} \frac{j!}{k_1! k_2! \cdots k_j! } \xi^{(k_1 + k_2 + \cdots + k_j)}(\eta(x)) \prod_{i=1}^j \left( \frac{\eta^{(i)}(x)}{i!} \right)^{k_i},
\]
\end{fact}

\begin{lemma}[Log-absolute monotonicity]
\label{cor: ck>0}
Let \(f\) be infinitely differentiable in a neighborhood of zero, with
\(f(0)=1\), and let \(g=\log f\).  For \(i\geq0\), define $a_i=\frac{f^{(i)}(0)}{i!}$ and $c_i=\frac{g^{(i)}(0)}{i!}$.

Suppose that the coefficients \(a_k\) are positive and log-convex, i.e. $$
a_{k+2} a_k \geq a_{k+1}^2 \qquad \text{for all $k \geq 0$}
$$

  Then \(c_k>0\) for every \(k\geq1\).
\end{lemma}
\begin{proof}
Here, we have the key identity
\begin{equation}
\label{kid}
f' = f g'
\end{equation}

We can use the Leibniz rule to differentiate both sides of (\ref{kid}), and evaluate at zero, to get: 
\[
f^{(k)}(0)
=
\sum_{j=1}^k
\binom{k-1}{j-1}
f^{(k-j)}(0) \, g^{(j)}(0).
\]
for any $k \geq 1$.  Dividing by \((k-1)!\) gives the recurrence
\begin{equation}
\label{eq:log-coefficient-recurrence}
k a_k
=
\sum_{j=1}^k j c_j a_{k-j}.
\end{equation}

We now show the claimed result by strong induction on \(k\). The case $k = 1$ holds since $c_1 = a_1 > 0$. For $k \geq 2$,   log-convexity gives
\[
a_{k-j}
\leq
\frac{a_k}{a_{k-1}}a_{k-1-j} \qquad \text{for $j < k$}
\]
Consequently, using \eqref{eq:log-coefficient-recurrence} at $k-1$, and using the induction hypothesis $c_0, \dots, c_{k-1} \geq 0$, gives
\begin{align*}
\sum_{j=1}^{k-1}j c_j a_{k-j}
&\leq
\frac{a_k}{a_{k-1}}
\sum_{j=1}^{k-1}j c_j a_{k-1-j} =
\frac{a_k}{a_{k-1}}(k-1)a_{k-1}
=
(k-1)a_k.
\end{align*}
Applying \eqref{eq:log-coefficient-recurrence} a second time at \(k\)  gives
\[
k c_k
=
k a_k-\sum_{j=1}^{k-1}j c_j a_{k-j}
\geq a_k>0.
\]
This completes the induction.
\end{proof}

\begin{proposition}\label{cor:monpower}
Let $f(x)$ be a function which is infinitely differentiable at zero with $f(0) = 1$, and let $h(x) = f(x)^q$ for  $q \geq 0$.   For $i \geq 0$, let $a_i = f^{(i)}(0)/i!$ and $b_i = h^{(i)}(0)/i!$.

Suppose that $f$ is log-absolutely monotonic, i.e. all derivatives of $\log f$ are nonnegative. Then $b_i \geq 0$ for all $i$, and the coefficients satisfy the recurrence
$$
b_0 = 1, \qquad b_k = \frac{1}{k} \sum_{j=1}^k \bigl( (q+1) j - k \bigr) a_j b_{k-j} \qquad \text{for $k \geq 1$}
$$
\end{proposition}
\begin{proof}
Let \(g(x)=\log f(x)\), so that $h(x)=\exp(qg(x)).$ Since \(g(0)=0\), \(q\geq0\), and all derivatives of \(g\) at zero
are nonnegative, the Faà di Bruno formula gives $h^{(k)}(0)\geq0$ for $k \geq 0$.

Differentiating \(h=f^q\) gives
\[
f(x)h'(x)=qf'(x)h(x).
\]
Apply the \((k-1)\)-st derivative to both sides and evaluate at zero.
By the Leibniz rule,
\[
\sum_{j=0}^{k-1}
\binom{k-1}{j}
f^{(j)}(0)h^{(k-j)}(0)
=
q\sum_{j=0}^{k-1}
\binom{k-1}{j}
f^{(j+1)}(0)h^{(k-1-j)}(0).
\]
Dividing by \((k-1)!\) and using $f^{(j)}(0)=j!a_j, h^{(j)}(0)=j!b_j$ 
gives
\[
\sum_{j=0}^{k-1}(k-j)a_jb_{k-j}
=
q\sum_{j=1}^k ja_jb_{k-j}.
\]
Since \(a_0=1\), isolating the term \(kb_k\) yields
\[
kb_k
=
\sum_{j=1}^k
\bigl((q+1)j-k\bigr)a_jb_{k-j}.
\]
Dividing by \(k\) proves the recurrence.
\end{proof}

\section{Numerical computation details: proof of \Cref{lem:correlation-factor}}
\label{app-numerical}
Here, we have parameters $r_1,r_2,g_1,g_2\in[0,1]$ where $r_2+g_2\le r_1$ and $r_1+g_1\le 1$. We define the following functions and expressions
$$
\alpha = 1.2337, \qquad \beta = 0.3
$$
$$
F(t) = \frac{0.68145}{\sqrt{1 - 0.53562 t}}, \qquad Q(t_0,\Delta t) = \int_{t = t_0}^{t_0 + \Delta_t} F(t) \dt 
$$
$$
y_i = Q(r_i,g_i), \qquad \rho_i = \alpha y_i, \qquad x_i =(1-\beta)F(r_i) g_i + \beta y_i
$$

Define 
\begin{align}
\label{fdefeqn}
f(r_1, r_2, g_1, g_2)  := \Bigl(\Psi( x_1, x_2; \rho_1, \rho_2) x_1 x_2 - \frac{\beta (y_1-F(r_1)g_1)y_2  }{Q(0,r_1)} \Bigr)/ \Bigl( F(r_1) g_1 y_2 \Bigr)
\end{align}

Our goal is to show that
$$
f(r_1, r_2, g_1, g_2) \leq c:= 0.3939
$$
for all $r_1, r_2, g_1, g_2$ in the region.  So we have a function maximization problem. We follow the standard procedure of partitioning the search space into small boxes of the form
$$
B = [r_1^{-},r_1^{+}]\times [r_2^{-},r_2^{+}]\times [g_1^{-},g_1^{+}]\times [g_2^{-},g_2^{+}]
$$
and upper bounding the function within each box via monotonicity properties.  This has three parts:
\begin{enumerate}
\item Within each box, we obtain bounds $x_i^{\pm}$ and $\lambda_i^-$ on $x_i, \lambda_i$.
\item We use these via \Cref{psicompbound} to obtain an upper bound $\Psi^+$ over the box.
\item We obtain an upper bound on $$
\Bigl(\Psi^+ x_1 x_2 - \frac{\beta (y_1-F(r_1)g_1)y_2  }{Q(0,r_1)} \Bigr)/ \Bigl( F(r_1) g_1 y_2 \Bigr) 
$$
within the box --- here, note that $\Psi$ is no longer regarded as a function of $x_i, \rho_i$.
\end{enumerate}

Throughout this section, it is convenient to define a function
$$
A(r,g)=\frac{Q(r,g)}{F(r)g},
$$
with \(A(r,0)=1\) by continuity. It is straightforward to check that, since $F(r)$ is positive increasing and log-convex, $A(r,g)$ is nondecreasing in both arguments. 

\begin{proposition}
Let $\lambda_i = \rho_i/x_i$. Within the feasible region of $B$, there holds $x_i \in [x_i^-, x_i^+]$ and $\lambda_i \geq \lambda_i^-$, for
\begin{align*}
x_i^- &= (1-\beta) F(r_i^-) g_i^- + \beta Q(r_i^-, g_i^-) \\
x_i^+ &= \begin{cases}
(1-\beta) F(r_i^+) g_i^+ + \beta  Q(r_i^+, g_i^+)  & \text{if $r_i^+ + g_i^+ \leq 1$}  \\
(1-\beta) F(r_i^*) (1 - r_i^*) + \beta  Q(r_i^*, 1-r_i^*)  & \text{if $r_i^+ + g_i^+ > 1$} 
\end{cases} \\
& \qquad \qquad \text{where $r_i^* = \max\{r_i^-, 1 - g_i^+ \}$} \\
\lambda_i^- &= \begin{cases}
\alpha Q(r_i^-, g_i^-) / x_i^- & \text{if $g_i^- > 0$} \\
\alpha & \text{if $g_i^- = 0$}
\end{cases}
\end{align*}
\end{proposition}
\begin{proof}
Since \(F\) is positive and increasing,
\[
\frac{\partial Q(r,g)}{\partial g}=F(r+g)>0,
\qquad
\frac{\partial Q(r,g)}{\partial r}=F(r+g)-F(r)\geq0.
\]
Thus both \(F(r)g\) and \(Q(r,g)\), and hence $x(r,g)=(1-\beta)F(r)g+\beta Q(r,g)$ are nondecreasing in \(r\) and \(g\).  This shows the bound $x_i^-$, and if  $r^+ + g^+ \leq 1$, it also shows the bound on $x_i^+$. If $r^+ + g^+ > 1$, then this shows that $x_i$ attains its maximum value at 
$$
(1-\beta) F(r)(1 - r) + \beta Q(r,1-r) 
$$
for some $r \in [ r_i^*, r_i^+]$. It can be checked that this expression is decreasing in $r$, so it is maximized at the left endpoint $r_i^*$.

For \(\lambda\), observe that
$$
\lambda(r,g) =\frac{\alpha A(r,g)}
       {1-\beta+\beta A(r,g)},
$$
which is  nondecreasing in \(A\).  Consequently \(\lambda_i\) is minimized at $r_i^-,g_i^-$, giving the stated lower bound.
\end{proof}

\begin{proposition}
\label{prop:objective-box-bound}
Suppose that
\[
\Psi(x_1,x_2;\rho_1,\rho_2)\leq \Psi^+
\qquad\text{throughout \(B\),}
\]
for some given scalar parameter $\Psi^+  \geq 0$.  Then, throughout the feasible part of \(B\), 
\begin{align*}
f(r_1,r_2,g_1,g_2)
\leq {}&
\Psi^+
\bigl(1-\beta+\beta A(\widehat r_1,g_1^-)\bigr)
\left(
  \beta+\frac{1-\beta}{A(r_2^-,g_2^-)}
\right) 
-\frac{\beta\bigl(A(\widehat r_1,g_1^-)-1\bigr)}
       {Q(0,\widehat r_1)} \\
       & \qquad \quad  \text{where $\widehat r_1=\min\{r_1^+,1-g_1^-\}$}
\end{align*}
\end{proposition}
\begin{proof}
Since the claimed bound is monotonic in $\Psi$, and we know that $\Psi \leq 1$ everywhere, it suffices to show this under the assumption that $\Psi^+ \leq 1$. Writing \(A_i=A(r_i,g_i)\), we have
\[
\frac{x_1}{F(r_1)g_1}=1-\beta+\beta A_1,
\qquad
\frac{x_2}{y_2}
=\beta+\frac{1-\beta}{A_2},
\qquad
\frac{y_1-F(r_1)g_1}{F(r_1)g_1}=A_1-1.
\]
and hence
$$
f(r_1, r_2, g_1, g_2) = \Psi(x_1, x_2; \rho_1, \rho_2) (1-\beta+\beta A_1)
\left(\beta+\frac{1-\beta}{A_2}\right)
-\frac{\beta(A_1-1)}{Q(0,r_1)}.
$$
Consequently, we have the convenient upper bound
\begin{equation}
\label{eq:objective-reduced}
f(r_1, r_2, g_1, g_2) 
\leq
\Psi^+(1-\beta+\beta A_1)
\left(\beta+\frac{1-\beta}{A_2}\right)
-\frac{\beta(A_1-1)}{Q(0,r_1)}.
\end{equation}

Since \(A(r,g)\) is nondecreasing in both arguments, 
the RHS of \eqref{eq:objective-reduced} is clearly
nonincreasing in \(r_2\) and \(g_2\). It is also nonincreasing in \(g_1\), since with \(r_1,r_2,g_2\)
fixed, its derivative with respect to \(A_1\) is
\[
\beta\left[
\Psi^+\left(\beta+\frac{1-\beta}{A_2}\right)
-\frac1{Q(0,r_1)}
\right]\
\]
which is nonpositive because \(A_2\geq1\), \(\Psi^+\leq1\), and
\(Q(0,r_1)\leq1\).

Finally, we claim that the RHS of \eqref{eq:objective-reduced} is nondecreasing in \(r_1\).  The term $\Psi^+(1-\beta + \beta A_1)(\beta + (1-\beta)/A_2)$ is clearly nondecreasing in $r_1$, since $A_1$ is nondecreasing in $r_1$ while $\Psi^+ \geq 0, A_2 \geq 1$. So it suffices to show that the final term
$$
\frac{A(r_1, g_1) - 1}{Q(0,r_1)}
$$
is nonincreasing in $r_1$. This can be checked mechanically since $A(r,g)$ and $Q(r,g)$ are both algebraically defined.

Overall, the RHS of \eqref{eq:objective-reduced} is nonincreasing in $r_2, g_2, g_1$ and nondecreasing in $r_1$, 
so we may decrease \(g_1,r_2,g_2\) to their lower endpoints
and increase \(r_1\) as far as feasibility permits.
\end{proof}

We divided the search space for each $r_i$ into boxes of size $\epsilon_r = 1/50$ and divided the search space for each $g_i$ into boxes of size $\epsilon_g = 1/1000$. \Cref{psicompbound} was used to obtain an $5$th-order upper bound on $\Psi$ in each box. For efficiency, we precomputed the $\kappa$ upper bounds, which only depend on the pair $(r_1, g_1)$ or $(r_2, g_2$) and not on the full 4-tuple $(r_1, r_2, g_1, g_2)$. The overall computation took about 3 days (wall-clock time) in Mathematica. 

Over the full search space, we found that the largest function value satisfied
$$
f(r_1, r_2, g_1, g_2) < 0.393825...
$$
as claimed.

\section{Derivation of function $F$ in \Cref{alg:online-rounding}}
\label{S-motivation-app}

Here we give the motivation for the function $F$. 
Let us suppose that all the edge demands $g_e$ are infinitesimal. We ignore all quadratic terms in $g$, and just focus on the first-order terms.  When we do this, we get $$
x_e \approx y_e \approx F(r_e) g_e, \qquad \rho_e \approx \alpha x_e
$$

We will leave $F$ as a function to be determined, and $\alpha$ as a free parameter. As a starting point, these should satisfy the properties in \cref{prop: numerical}:
\begin{enumerate}
\item[(a)] $F$ is nonnegative, increasing, and concave-up;
\item[(b)] $\int_{t=0}^1 F(t) \dt = 1/\alpha$;
\item[(c)] $F(1) = 1$.
\end{enumerate}

(Note that, in \cref{prop: numerical}, $\alpha \int_{t=0}^1 F(t) \mathrm{d}t$ and $F(1)$ are very slightly less than $1$; these are artifacts caused by our use of decimal values for the numerical constants.)

By \Cref{poisson-prop} and \Cref{tre2}(3), we can approximate the $\Psi$ function for edges $e,f$ by:
$$
\Psi(x_e, x_f; \rho_e, \rho_f) \approx \binom{2 \alpha}{\alpha}^{-1} 
$$
where we again ignore all high-order terms. By the same argument as in \cref{thm: main apx},  the probability that edge $e$ is matched is then at least
$$
g_e F(r_e) \Bigl( 1 - c \int_{t=0}^{r_e}F(t) \dt \Bigr) \qquad \text{for constant $c = \binom{2 \alpha}{\alpha}^{-1}$}. 
$$

Define $Q(r) = \int_{t=0}^r F(t) \dt$, so the probability that $e$ is matched is $g_e Q'(r_e) (1 - c Q(r_e))$. We want an approximation ratio irrespective of $r_e$; hence, the quantity $Q'(r_e) (1 - c Q(r_e))$ should not depend on $r_e$:
$$
\frac{\mathrm{d}}{\mathrm{d} t} \Bigl[ Q'(t) (1 - c Q(t)) \Bigr] = 0
$$

So we have a second-order differential equation
$$
Q''(t) - c (Q'(t))^2 - c Q(t) Q''(t) = 0
$$
along with boundary conditions $Q(0) = 0, Q'(1) = 1, Q'(t) > 0$, which has the solution
\begin{equation}
\label{afat}
Q(t) = \frac{1 - \sqrt{1 - t \bigl( 2 c \sqrt{c^2+1} -2 c^2  \bigr) }}{c}
\end{equation}

At this point, we can choose $\alpha = 1.2337$ (with corresponding value of $c$) in order to make $\alpha Q(1)$ slightly smaller than one. We can also replace the constant terms in (\ref{afat}), which are functions of $\alpha$, with nearby numerical values to give $Q'(1) < 1$ and $Q(1) < 1/\alpha$. We finally define $F(t) = Q'(t)$.

    \section{Proof of \Cref{thm: sch_apx}}\label{ap: sch}
We complete the verification of \Cref{thm: sch_apx} by following the scheduling framework of \cite{harris2024dependent}, with updated correlation bounds from using the Dirichlet rounding mechanism. We state the main intermediate results and omit routine numerical calculations and technical details that follow directly from the reference.

\noindent \textbf{Note:  For the purposes of this section, the numerical analysis and calculations are non-rigorous.}

\subsection{Upper bound on $\E[Z]$}\
As the first step, we use \Cref{xjjprop} and write the upper bound in terms of the \textit{bonus terms} $B_{k,\ell}$:
\begin{align}\label{eq: upper + bonus terms}
    E[Z|P_{\text{offset}}]\leq Q+\frac{L^2}{2} - \sum_{k,\ell} P_k^2B_{k,\ell},
\end{align}
where $$
 B_{k,\ell} = \frac{1}{2} \Bigl(
 \sum_{j\in \mathcal{C}^*_{k,\ell}}x_j^2 H_j^2 + \sum_{\substack{j,j'\in \calC_{k,\ell}^* \\j\neq j'}} x_j x_{j'} H_jH_{j'} \bigl( 1 - \Psi(x_j, x_{j'}, \rho_j, \rho_{j'}) \bigr) 
 \Bigr).
 $$
\cref{eq: upper + bonus terms} enables us to focus on obtaining a lower bound for the expectation of $B_{k,\ell}$.

\begin{lemma}[cf Lemma 24 of \cite{harris2024dependent}]
\label{wclem}
For a cluster $\mathcal C_{k,\ell}^{*}$, let $\mathcal U^* = \mathcal U_{k,\ell} \cap \mathcal J^*$ denote the set of untruncated jobs in the cluster,  and define parameters
$$
\lambda = \min\{ \tau, x(\mathcal C_{k,\ell}) \}, \quad a = 1 - \frac{1}{\binom{2/\lambda}{1/\lambda}}, \quad r = \sum_{j \in \mathcal U^*} x_j, \quad s = \sum_{j \in \mathcal U^*} x_j H_j
$$

(Note that the definition of $\lambda$ is in terms of the full cluster $\mathcal C_{k,\ell}$, not the partial cluster $\mathcal C^{*}_{k,\ell}$)

If $\mathcal C_{k,\ell}^{*}$ has no truncated job, then 
$$B_{k,\ell}\geq\frac{as^2}{2}.$$

If $\mathcal C_{k,\ell}^{*}$ has a truncated job $j'$ with $y = x_{j'}, d = x_{j'} H_{j'}$, then
$$
B_{k,\ell} \geq  \frac{a s^2}{2} + \frac{d^2}{2} +  \frac{s}{2}  \inf_{x\in (0,r]} \Bigl(x(1-a)+2d (1- \Psi(x,y; x/\tau,  1 - r/\tau)) \Bigr).
$$
\end{lemma}
\begin{proof} 
Let $\mathcal{T}^* = \mathcal T_{k,\ell} \cap \mathcal{J}^*$ denote the set of truncated jobs in the cluster $\mathcal C_{k,\ell}^{*}$. Since $|\mathcal{T}^*| \leq 1$, we can expand the formula for $B_{k,\ell}$ to get:
\begin{align*}
2B_{k,\ell} &= \sum_{j\in \mathcal{T}^*}x_j^2 H_j^2 + \sum_{j\in \mathcal{U}^*} \Bigl( x_j^2 H_j^2 +\sum_{j'\in \calU^*:j'\neq j} \bigl( x_j x_{j'} H_j H_{j'} (1- \Psi(x_j, x_{j'}, \rho_j, \rho_{j'}) \bigr)  \\
& \qquad +    2\sum_{j' \in \mathcal{T}^*}  x_j x_{j'} H_j H_{j'} (1 - \Psi(x_j, x_{j'}, \rho_j, \rho_{j'})) \Bigr)
\end{align*}

Here $\rho_j = x_j/\lambda$ for each untruncated job. Based on non-rigorous numerical analysis of the $3^{\text{rd}}$-order approximation of $\Psi$ (specifically, using \Cref{mainestlemma} with $k_1,k_2 = 3$) over range $\lambda\in[0,\tau]$ and  $x_1,x_2$ with $x_1 + x_2 \leq \lambda$, the following bound appears to hold:
$$
\Psi(x_1, x_2, x_1/\lambda, x_2/\lambda) \leq  \binom{2/\lambda}{1/\lambda}^{-1} = 1-a.
$$

Now consider a job $j \in \mathcal{U}^*$. We can bound the sum over $j' \in \cal U^*$ as: 
\begin{align*}
&\sum_{j'\in \calU^*:j'\neq j} \bigl( x_j x_{j'} H_j H_{j'} (1 - \Psi(x_j, x_{j'}, \rho_j, \rho_{j'})) \geq \sum_{j'\in \calU^*:j'\neq j} a  x_j x_{j'}  H_j H_{j'} = a x_j H_j s - a x_j^2 H_j^2
\end{align*}

Thus, if $\mathcal T^* = \emptyset$, we sum over all jobs to get:
\begin{align*}
2B_{k,\ell} &\geq \sum_{j\in \mathcal{U}^*} \bigl( x_j^2 H_j^2 + a x_j H_j s - a x_j^2 H_j^2 \bigr) \geq \sum_{j\in \mathcal{U}^*_{k,\ell}} x_j H_j \bigl( a s \bigr) = a s^2
\end{align*}
where the last inequality holds since $a \leq 1$ and $H_j \geq 1$. 

Otherwise, if there is a truncated job $j'$, then $\rho_{j'} = 1-r/\tau$ and $\lambda = \tau$ and $\rho_j = \frac{x_j}{\tau}$. This job $j'$ has 
$$
x_j x_{j'}  (1 - \Psi(x_j, x_{j'}, \rho_j, \rho_{j'})) H_j H_{j'}  = d x_j H_j \bigl( 1 - \Psi( x_j, y, x_j/\tau, 1 - r/\tau) \bigr)
$$
and then summing over all jobs $j \in \mathcal U^*$ gives:
\begin{align*}
2B_{k,\ell} &\geq d^2 + \sum_{j\in \mathcal{U}^*} \Bigl( x_j^2 H_j^2 + a x_j H_j s - a x_j^2 H_j^2 + 2 x_j H_j d \bigl( 1 - \Psi(x_j, y, x_j/\tau, 1 - r/\tau) \bigr) \Bigr) \\
&\geq d^2 + \sum_{j\in \mathcal{U}^*} x_j H_j \Bigl( x_j H_j + a s - a x_j H_j  +  2 d \bigl( 1 - \Psi( x_j, y, x_j/\tau, 1 - r/\tau) \bigr) \Bigr) \\
&\geq d^2 + \sum_{j\in \mathcal{U}^*} x_j H_j \Bigl( a s + (1-a) x_j +  2 d \bigl( 1 - \Psi( x_j, y, x_j/\tau, 1 - r/\tau) \bigr) \Bigr)
\end{align*}
The result now follows since $\sum_{j \in \mathcal U^*} x_j H_j = s$ and each term $(1-a) x_j + 2  d ( 1 - \Psi( x_j, y, x_j/\tau, 1 - r/\tau)) $ is at least $\inf_{x\in (0,r]} \bigl(x(1-a)+2 d (1- \Psi(x, y; x/\tau,  1 - r/\tau)) \bigr)$.
\end{proof}

This gives two main estimates for cluster bonuses:
\begin{proposition}\label{prop:non-leftover}
Under the assumptions of \Cref{wclem}, a non-leftover cluster $C^*_{k, \ell}$ satisfies
    \begin{align}\label{eq: c1}
        B_{k,\ell}\geq c_1\sum_{j\in \calC_{k,\ell}^*}x_j(H_j-\kappa)\qquad\text{for $c_1=0.684$ and $\kappa = 0.778.$}
    \end{align}
\end{proposition}
\begin{proof}
Suppose $\calC_{k,\ell}^{*}$ has no truncated job. Then $\sum_{j\in \calC_{k,\ell}^*}x_j(H_j-\kappa) = s - \kappa r$ and $B_{k,\ell}\geq\frac{as^2}{2}$ with $a = 1- \frac{1}{\binom{2/r}{1/r}}$ and $r \geq \theta$. Let $t = s/r$; note that since $H_j \in [1, \pi]$ we also have $t \in [1,\pi]$. We have:
\begin{align*}
\frac{B_{k,\ell}}{s - \kappa r} \geq \frac{(1 - \frac{1}{\binom{2/r}{1/r}}) s^2}{2(s - \kappa r)} = \frac{(1 - \frac{1}{\binom{2/r}{1/r}}) r t^2}{2(t - \kappa)}
\end{align*}

By straightforward calculus, we see that $\frac{t^2}{2 (t - \kappa)} \geq 2 \kappa$. Also, the expression $(1 - \frac{1}{\binom{2/r}{1/r}}) r$ is increasing in $r$, so it is at least $(1 - \frac{1}{\binom{2/\theta}{1/\theta}}) \theta$. Overall, we have
\begin{align*}
\frac{B_{k,\ell}}{s - \kappa r} \geq  (1 - \frac{1}{\binom{2/\theta}{1/\theta}}) \theta) \cdot 2 \kappa = 0.68532025702558509220.... \geq c_1.
\end{align*}

Next, suppose $\calC_{k,\ell}^{*}$ has a truncated job. Then $\sum_{j\in \calC_{k,\ell}^*}x_j(H_j-\kappa) = s + d - \kappa(r + y)$. By \Cref{wclem}, 
\begin{align*}
    \frac{B_{k, \ell}}{s + d - \kappa(r + y)} \geq \inf{\substack{r \in [0, \theta],x\in(0,r]\\ y \in [\tau -r, 1] \\  s \in [r, r\pi], d \in [y, y\pi]}} \frac{\frac{as^2}{2} + \frac{d^2}{2} + \frac{s}{2} \left( x(1-a)+2d\left( 1 - \Psi(x, y; x/\tau,  1 - r/\tau)\right)\right)}{s + d - \kappa(r + y)}
\end{align*}
Using the $3^{rd}$-order approximation for $\Psi$, 
it is shown numerically to be at least $0.68707 > c_1$. 
\end{proof}

\begin{proposition}\label{cor:leftover}
Under the assumptions of \Cref{wclem}, a leftover cluster $\mathcal C_{k,\ell}^*$ satisfies $$
B_{k,\ell} \geq c_2 \Bigl( \sum_{j \in \calC_{k,\ell}^*} x_j H_j \Bigr)^2
\qquad \text{ for constant $c_2 = 0.3747129$} 
$$
\end{proposition}
\begin{proof}
Let parameters $\lambda,a,r,s,y,d$ be as in \Cref{wclem}.
The leftover cluster cannot contain a truncated job in $\mathcal C^*_{k,\ell}$. So $d = y = 0$ and $s = \sum_{j \in \mathcal C^*_{k,\ell}} x_j H_j$. Also note that $\lambda \leq \tau$ and hence $a \geq 1 - \binom{2/\tau}{1/\tau}^{-1}$. We have
$$
B_{k,\ell} \geq \Bigl( 1 - \frac{1}{\binom{2/\tau}{1/\tau}} \Bigr) \cdot \frac{ \bigl( \sum_{j \in \calC_{k,\ell}^*} x_j H_j \bigr)^2 }{2}
$$
and observe that $\bigl( 1 - \frac{1}{\binom{2/\tau}{1/\tau}} \bigr) \cdot \frac{1}{2} =0.374712911362113\ldots\ge c_2$.
\end{proof}    
    
\begin{lemma}[cf Proposition 26 of \cite{harris2024dependent}]\label{lem: c3}
Let
\[
R_k := \sum_{j \in \mathcal{C}_{k,\mathrm{left}}^*} x_j,
\qquad
S_k := \sum_{j \in \mathcal{C}_{k,\mathrm{left}}^*} x_j H_j,
\qquad
T_k := \frac{S_k}{R_k},
\]
and define
\[
f(r,t) := \max\{0,\; c_1 (t - \kappa) - c_2 r t^2\},
\qquad
D := \sum_k P_k^2 R_k f(R_k, T_k).
\]
Then
\[
\E[Z] \le \E[D] + c_3 Q + \frac{L^2}{2},
\qquad
\text{where } c_3 := 0.814462.
\]
\end{lemma}
\begin{proof}
Summing over clusters in $\calP_k^*$, and using \Cref{prop:non-leftover} or \Cref{cor:leftover} for each cluster, gives: 
\[
\sum_{\ell} B_{k,\ell}
\;\ge\;
- R_k f(R_k,T_k) + 
c_1 \sum_{j \in \mathcal{P}_k^*} x_j (H_j - \kappa).
\]
Substituting this bound into \eqref{eq: upper + bonus terms} gives
\[
\E[Z \mid P_{\mathrm{offset}}]
\le
Q + \frac{L^2}{2}
+ D
- c_1 \sum_j x_j p_j^2 \frac{H_j - \kappa}{H_j^2}.
\]

Taking expectation over the random variable $P_{\mathrm{offset}}$,
\[
\E\!\left[(H_j - \kappa)/H_j^2\right]
=
\frac{1}{\log \pi} \int_{1}^{\pi} \frac{h - \kappa}{h^3}\,\mathrm{d}h
=
\frac{\kappa - 2\pi + 2\pi^2 - \kappa \pi^2}{2 \pi^2 \log \pi}.
\]
Therefore,
\[
\E[Z]
\le
\E[D]
+
\Bigl(
1
-
c_1 \cdot
\frac{\kappa - 2\pi + 2\pi^2 - \kappa \pi^2}{2 \pi^2 \log \pi}
\Bigr) Q
+
\frac{L^2}{2}.
\]

Then
\[
c_3 \geq
1
-
c_1 \cdot
\frac{\kappa - 2\pi + 2\pi^2 - \kappa \pi^2}{2 \pi^2 \log \pi} 
=0.814461971070319\ldots\
\]
completes the proof.
\end{proof}

\subsection{Bound on $\mathrm{LB}$}
The lower bound derivation reuses the structural inequalities of \cite{harris2024dependent}, with slightly different numerical constants. We therefore state the main results and omit the corresponding numerical calculations.

\begin{lemma}\label{lem: sch tech}
For parameters \(\beta := 1.93\) and \(\kappa, c_1, c_2, c_3\) as above, define
\[
g(r, t, h) := h \cdot \Biggl( 1 - \sqrt{ 1 - \frac{\beta (h-\kappa) f(r, t)}{h^2 (t-\kappa)} } \Biggr), \qquad
g_k(h) := g(R_k, T_k, h).
\]
\[
A := \sum_k P_k \sum_{j \in \mathcal C^{*}_{k,\mathrm{left}}} x_j g_k(H_j), \qquad  q := Q + \E[D]/c_3
 \]
 
Then, for a fixed parameter $\gamma := 0.00594$, we have
\begin{enumerate}
\item[(i)] (cf Corollary~21 of \cite{harris2024dependent}) \ \ \ \ $\mathrm{LB} \ge \max\!\left\{Q,\; \tfrac12(Q + L^2)\right\}.$

\item[(ii)] (cf Lemma~28 of \cite{harris2024dependent}) \ \ \ \ \ \ $\mathrm{LB} \ge \tfrac12\bigl(Q + \beta D + (L - A)^2\bigr).$

\item[(iii)] (cf Lemma~30 of \cite{harris2024dependent}) \ \ \ \ \ \ \thinspace  $\E[A^2] \le \gamma \E[D] + c_5 Q$ where $c_5 = 0.004833.$

\item[(iv)] (cf Proposition~33 of \cite{harris2024dependent}) \ \ \ $\mathrm{LB}
\;\ge\;
\frac{\beta c_3 q + \max\{0,\; L - c_6 \sqrt{q}\}^2}{\beta c_3 + 1}$ where  $c_6 := 0.0695551$.
\end{enumerate}
\end{lemma}
\bigskip

\begin{theorem}\label{thm: sch_apx1}[Restated]
We have $\E[Z] \le 1.387 \cdot \mathrm{LB}.$
\end{theorem}
\begin{proof}
    By \Cref{lem: c3} and \Cref{lem: sch tech}(iv),
    \[
\frac{\E[Z]}{\mathrm{LB}}
\le
\frac{(\beta c_3 + 1)\bigl(c_3 q + L^2/2\bigr)}
     {\beta c_3 q + \max\{0,\, L - c_6 \sqrt{q}\}^2}.
\]
    The right-hand side can be shown to be at most $1.38695$ using symbolic maximization.
\end{proof}

\bibliographystyle{alpha}
\bibliography{refs}
\nocite{*}
\end{document}